\documentclass{amsart}
\usepackage{amssymb, amsmath, graphicx, amsthm,amsfonts}
\usepackage{array}        
\usepackage[utf8]{inputenc}

\usepackage{tikz}
\usepackage{tikz-cd}
\usetikzlibrary{shapes.geometric}
\usetikzlibrary{cd}
\usetikzlibrary{graphs,decorations.pathmorphing,decorations.markings}
\usepackage{soul}
\usepackage{bm}

\usepackage{mathtools}

\usepackage{enumerate} 
\usepackage{tasks}

\newcommand{\secmath}[1]{\texorpdfstring{$ #1 $}{T}}

\usepackage[colorlinks = true,
            linkcolor = blue,
            urlcolor  = red,
            citecolor = red,
            anchorcolor = blue,
            urlbordercolor= red ]{hyperref}

\usepackage{changepage}

\usepackage{caption}
\usepackage{subcaption}
\newtheorem{theorem}{Theorem}[section]

\newtheorem{prop}[theorem]{Proposition}
\newtheorem{remark}[theorem]{Remark}

\DeclareMathOperator{\Aut}{Aut}

\DeclareMathOperator{\diag}{diag}
\DeclareMathOperator{\id}{Id}

\DeclareMathOperator{\GL}{GL}
\DeclareMathOperator{\SL}{SL}
\DeclareMathOperator{\PGL}{PGL}

\DeclareMathOperator{\Tr}{Tr}
\DeclareMathOperator{\sgn}{sgn}

\newcommand{\pround}[1]{\left( #1 \right)}
\newcommand{\psquare}[1]{\left[ #1 \right]}
\newcommand{\pbrace}[1]{\left\{ #1 \right\} }
\newcommand{\pangle}[1]{\left\langle #1 \right\rangle}
\newcommand{\abs}[1]{\left \lvert #1 \right \rvert}

\usepackage{caption}

\subjclass[2020]{}
\usepackage{array}
\newcolumntype{P}[1]{>{\centering\arraybackslash}p{#1}}

\usepackage{wasysym}

\usepackage{float}

\usepackage{longtable}

\usepackage{xcolor}
\definecolor{lime}{HTML}{A6CE39}
\DeclareRobustCommand{\orcidicon}{%
	\begin{tikzpicture}
	\draw[lime, fill=lime] (0,0) 
	circle [radius=0.16] 
	node[white] {{\fontfamily{qag}\selectfont \tiny ID}};
	\draw[white, fill=white] (-0.0625,0.095) 
	circle [radius=0.007];
	\end{tikzpicture}
	\hspace{-2mm}
}

\foreach \x in {A, ..., Z}{%
	\expandafter\xdef\csname orcid\x\endcsname{\noexpand\href{https://orcid.org/\csname orcidauthor\x\endcsname}{\noexpand\orcidicon}}
}


\title{Towards a Classification of Charge-3 Monopoles with Symmetry}

\author[H. W. Braden]{H. W. Braden \orcidA{}}
\author[Linden Disney-Hogg]{Linden Disney-Hogg \orcidB{}}

\address{
School of Mathematics and Maxwell Institute for Mathematical Sciences\\ The University of Edinburgh\\ 
Edinburgh EH9 3FD, Scotland, U.K.
}

\email{hwb@ed.ac.uk}
\email{A.L.Disney-hogg@sms.ed.ac.uk}

\thanks{\textbf{Acknowledgements.} We are grateful to Conor Houghton for correspondence and to Paul Sutcliffe for discussions and providing code (discussed below) that we have modified to plot the accompanying monopole solutions. The research of LDH is supported by a UK Engineering and Physical Sciences Research Council (EPSRC) studentship.}

\thanks{\textbf{Data Availability.} The datasets generated during the current study and the code for their creation/analysis are available from the second author upon reasonable request.}

\thanks{\textbf{Statements and Declarations.} The authors have no relevant financial or non-financial interests to disclose.}

\thanks{EMPG-23-06}

\begin{document}

\dedicatory{Dedicated to Nigel Hitchin: for the 40th anniversary of his construction of monopoles \\ - and much more.}

\maketitle

\begin{abstract}
We classify all possible charge-3 monopole spectral curves with non-trivial automorphism group and within these identify those with elliptic quotients. By focussing on elliptic quotients  the transcendental 
constraints for a monopole spectral curve become ones regarding periods of elliptic functions.
We construct the Nahm data and new monopole spectral curves with $D_6$ and $V_4$ symmetry, the latter based on an integrable complexification of Euler's equations, and for which energy density isosurfaces are plotted. Extensions of our approach to higher charge and hyperbolic monopoles are discussed.

\end{abstract}

\section{Introduction}

In this letter we shall classify and construct new charge-$3$ $\operatorname{SU}(2)$ monopoles; throughout we will work just with this gauge group. To describe this work further we first recall some necessary background.
Following the success of the ADHM construction of instanton solutions to the self-dual Yang-Mills equations Nahm introduced his eponymous equations to solving the reduction of the self-dual Yang-Mills equations $F = \ast D\phi$ which describes monopoles; alternately these reduced equations arise as the BPS limit of a Yang-Mills-Higgs system.
Here a (Euclidean BPS) monopole is the data $(A, \phi)$ where $A$ is a connection of a $\operatorname{SU}(2)$ principal bundle over $\mathbb{R}^3$ with associated curvature $F$, and $\phi$ (the Higgs field) is a section of the adjoint bundle. These satisfy the equation $F = \ast D\phi$ as well as the boundary conditions that as $r \to \infty$
(i)
 $\abs{\phi} = 1 - \frac{k}{2r} + \mathcal{O}(r^{-2})$,
where $\abs{\phi}^2=-\frac12 \Tr(\phi^2)$ is the norm coming from the Killing form on $\mathfrak{su}(2)$ and $r$ is the Euclidean distance from origin; (ii)
 $\frac{\partial\abs{\phi}}{\partial\Omega} = \mathcal{O}(r^{-2})$ and $\Omega$ is a solid angle;
(iii)   $\abs{D\phi} = \mathcal{O}(r^{-2})$.
The (positive) integer $k$ fixes the topological class of the data and is called the charge. Nahm's modification of the ADHM
construction then sought a triplet of $k\times k$ matrices 
of a real parameter $s\in [0,2]$ such that these satisfied (i) the Nahm equations
\begin{equation}\label{nahmseqns}
 \frac{dT_i}{ds} = \frac{1}{2} \sum_{j,k=1}^3 \epsilon_{ijk} [T_j,T_k],  
\end{equation}
together with (ii) the $T_i(s)$ are regular for all $s \in (0,2)$
and with simple poles at $s=0,2$ if $k>1$, whose residues form an irreducible $k$-dimensional representation of $\operatorname{SU}(2)$;  and (iii) $T_i(s) = -T_i^\dagger(s), \, T_i(s) = T^T_i(2-s)$.
Hitchin in his seminal work \cite{Hitchin1983} introduced yet a third description of monopoles, the \emph{spectral curve}, and proved the equivalence of all three descriptions. Here the
spectral curve $\mathcal{C} \subset T\mathbb{P}^1 \overset{\pi}{\to}\mathbb{P}^1 $ is a compact algebraic curve with no multiple components of genus $(k-1)^2$ such that (i) $\mathcal{C}$ is real with respect to a (to be given) anti-holomorphic involution $\tau$; (ii) there is a family of line bundles
$\mathcal{L}\sp{s}$ on $\mathcal{C}$ such that $\mathcal{L}^2$ is trivial and $\mathcal{L}(k-1):=\mathcal{L}\otimes \pi\sp\ast\mathcal{O}_{\mathbb{P}^1}(k-1)$ is real; and (iii)
$H^0(\mathcal{C}, \mathcal{L}^s(k-2))=0$ for all $s \in (0,2)$. 
Introducing coordinates $\zeta, \eta$ on $T\mathbb{P}^1$ corresponding to the base $\mathbb{P}^1$ and fibre respectively
then $\mathcal{L}\sp{s} \to T\mathbb{P}^1$ (and by restriction, to $\mathcal{C}$) is the (holomorphic) line bundle defined by $\exp(s\eta/\zeta)$ with $\eta/\zeta \in H^1(T\mathbb{P}^1, \mathcal{O})$.
 Note that $T\mathbb{P}^1$ arises here because it is the mini-twistor space of oriented geodesics in Euclidean 3-space
 \cite{Hitchin1982}. The action of $\tau$ is $(\zeta, \eta) \mapsto (-1/\bar{\zeta}, -\bar{\eta}/\bar{\zeta}^2)$, and so a generic spectral curve satisfying the reality constraints may be written as the zero set of a polynomial 
\[
0 = P(\zeta, \eta) := \eta^k + \sum_{r=1}^{k} p_{2r}(\zeta) \eta^{k-r}, \quad  p_{2r}(\zeta) = (-1)^r \zeta^{2r} \overline{p_{2r}(-1/\bar{\zeta})}, 
\]
where $p_{2r}$ is a polynomial of degree $2r$ in $\zeta$. In what follows we shall refer to \emph{Nahm data} as matrices
$\{T_i\}$ satisfying Nahm's three constraints and a \emph{monopole spectral curve} as a curve $\mathcal{C}$  satisfying Hitchin's three constraints.

Our understanding of monopoles, the self-duality equations
and Nahm's equation have developed greatly in the 40 years since
\cite{Hitchin1983}. The moduli space of monopoles of a given charge has attracted much attention, and a rational map description \cite{Donaldson1984} allows different insights and facilitates numerical solutions; integrable systems techniques have also been brought to the fore \cite{Ercolani1989, Braden2018}. Yet despite this progress very few monopole spectral curves have been found in the intervening period owing to the transcendentality of the Hitchin conditions (see \cite{Braden2021}). While monopoles of charge 1 and 2 are well-understood (for a review, see \cite{Braden2021a}) little progress has been made for higher charges. In all cases known the use of symmetry to simplify the conditions has been required; in nearly all of these we may quotient by the group of symmetries to an elliptic curve. Motivated by this history our first result here is to classify all possible charge-3 monopole spectral curves by their automorphism groups and within these identify those with elliptic quotients. From these we will then construct new monopole spectral curves with $D_6$ and $V_4$ symmetry. In section \S\ref{sec: classifying curves by automorphism group} we prove:

\begin{theorem}\label{charge3auto}
    Let $\mathcal{C} \subset T\mathbb{P}^1$ be a charge-3 monopole  spectral curve with $H \leq \Aut(\mathcal{C})$ such that the quotient genus $g(\mathcal{C}/H) = 1$. Then, up to an automorphism of $T\mathbb{P}^1$, the curve is given by the vanishing of one of the following 5 forms:
    \begin{tasks}[style=enumerate](1)
        \task $\eta^3 + \eta[(a+ib) \zeta^4 + c \zeta^2 + (a-ib)] + [(d+ie)\zeta^6 + (f+ig) \zeta^4 - (f-ig) \zeta^2  - (d-ie)]$,
        \task $\eta^3 + \eta[a(\zeta^4 + 1) + b\zeta^2] + ic\zeta(\zeta^4-1)$,
        \task $\eta^3 + a \eta \zeta^2 + ib\zeta(\zeta^4 - 1)$, 
        \task $\eta^3 + a\eta \zeta^2 + b(\zeta^6 - 1)$, 
        \task $\eta^3 + ia\zeta(\zeta^4 - 1)$, 
    \end{tasks}
    where $a, b, c, d, e, f, g \in \mathbb{R}$. 
\end{theorem}
Our result does not itself guarantee the existence of monopole spectral curves in these families, and the classes intersect (for example, 5 is a special case of 3). In previous works monopole spectral curves of the form 3 and 5 have been understood in \cite{Houghton1996c} and \cite{Hitchin1995} as corresponding to charge-3 twisted line scattering and the tetrahedrally-symmetric monopole respectively, while one special case of the form 2 was understood in \cite{Houghton1996b} as the class of inversion-symmetric monopoles, with another in \cite{Hitchin1983} as the axially-symmetric 3-monopole. Curves of the form 2 had been observed in \cite[(3.71)]{Houghton1997}, but the Hitchin constraints were only imposed for a restricted subset. In Theorem \ref{thm: D6 monopole spectral curves} we  determine the general monopole spectral curve and Nahm data in class 4 and in Theorem \ref{thm: V4 monopole spectral curves} the same for class 2. This provides the necessary data in order to plot energy density isosurfaces following \cite{Houghton1996}\footnote{We are grateful to Paul Sutcliffe for providing us with the initial code which we modified to make our plots.}; Figure \ref{fig: two pictures of a specific monopole} gives for example a previously unknown $V_4$ configuration. Along with the class of charge-3 monopoles described via an implicit condition in \cite{Braden2011a}, these form all the charge-3 monopole spectral curves currently known, which fit together as shown in Figure \ref{fig: charge-3 curves} for some parameter values. Figure \ref{fig: charge-3 curves automorphism groups.} shows the relations between the symmetry groups of the curves. 

\begin{figure}
    \centering
     \begin{subfigure}[c]{0.49\textwidth}
         \centering
         \includegraphics[width=\textwidth]{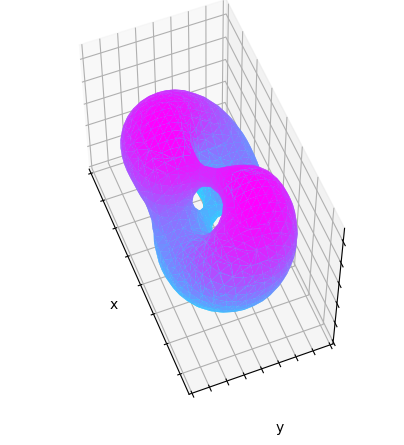}
         \caption{}
     \end{subfigure}
     \begin{subfigure}[c]{0.49\textwidth}
         \centering
         \includegraphics[width=\textwidth]{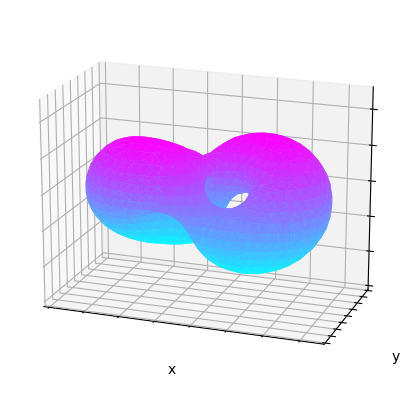}
         \caption{}
     \end{subfigure}
     \caption{Surface of constant energy density $\mathcal{E}=0.18$ for the $V_4$ monopole given by the parameters (see Theorem \ref{thm: V4 monopole spectral curves}) $m=0.6$, $\alpha = -2.0$, $\sgn = 1$} 
     \label{fig: two pictures of a specific monopole}
\end{figure}

Our approach is as follows. In section \S\ref{sec: classifying curves by automorphism group} of the paper we will prove Theorem \ref{charge3auto}. Once we have the automorphism groups of interest we will take the procedure introduced in \cite{Hitchin1995} and developed in \cite{Houghton1996,Houghton1996b,Houghton1996c} and apply this to the relevant symmetry. This procedure is recalled in Appendix \ref{sec: initial computation of nahm matrices} where Nahm's equations for case 4 are reduced to the Toda equations before further reduction is described in the text. Similarly in Appendix \ref{sec: v4 nahm matrices} Nahm's equations for the $V_4$ symmetric case are determined. This yields a complex extension of the Euler equations. We then show how these equations are solved in terms of elliptic functions on the quotient elliptic curve, first in Section \S\ref{D6 Monopoles} for the $D_6$-symmetric monopole and then in Section \S\ref{V4 Monopoles} for the $V_4$-symmetric monopole. Here the rationale for focussing on elliptic quotients is most evident: the transcendental constraints implicit in the works of Hitchin and Ercolani-Sinha become ones regarding periods of elliptic functions. We relegate to Appendix \ref{sec: hypergeometric function proofs} a number of properties of elliptic and related functions used in the text and proofs of some statements requiring these. We will not deal with the remaining $C_2$-symmetric case here discussing this further in Section \S\ref{Conclusion} which is a conclusion.

\begin{center}
\begin{figure}[H]
    \begin{tikzcd}
     & \eta^3 + \alpha_2 \eta \zeta^2 + \alpha_3 \zeta^3 + \beta(\zeta^6-1), \text{\cite{Braden2011a}} \arrow[dl, "\alpha_3=0"'] \arrow[d, "\alpha_2=0 \text{ and rotation}"] \\
     \eta^3+  \alpha_2 \eta \zeta^2  + \beta(\zeta^6-1), \text{here} \arrow[d, "\beta=0"'] & \eta^3 + c\zeta(\zeta^4-1), \text{\cite{Hitchin1995}} \\
     \eta[\eta^2 + \pi^2  \zeta^2], \text{\cite{Hitchin1983}} & \eta^3 + b\eta \zeta^2 + c \zeta(\zeta^4-1), \text{\cite{Houghton1996c}} \arrow[u, "b=0"'] \arrow[l, "c=0"] \\
     \eta\psquare{\eta^2 + a(\zeta^4+1) + b\zeta^2}, \text{\cite{Houghton1996b}} \arrow[u, "a=0"] & \eta^3 + \eta[a(\zeta^4 + 1) + b \zeta^2] + c\zeta(\zeta^4-1), \text{here}. \arrow[u, "a=0"] \arrow[l,"c=0"]
    \end{tikzcd}
    \caption{Known charge-3 spectral curves and their relations. We do not specify the constraints on the parameters}
    \label{fig: charge-3 curves}
\end{figure}
\end{center}

\begin{center}
\begin{figure}[H]
    \begin{tikzcd}
     & S_3 \arrow[dl, "\alpha_3=0"'] \arrow[d, "\alpha_2=0 \text{ and rotation}"] \\
     S_3, C_6 \leq D_6 \arrow[d, "\beta=0"'] & A_4 \leq C_3 \times S_4 \\
     C_2 \times \operatorname{SO}(2) & C_4 \leq D_4 \arrow[u, "b=0"'] \arrow[l, "c=0"] \\
     V_4 \arrow[u, "a=0"] & V_4 \leq V_4. \arrow[u, "a=0"] \arrow[l,"c=0"]
    \end{tikzcd}
    \caption{Automorphism groups of known charge-3 spectral curves and their relations, presented as $G$ or $H \leq G$ where $G$ is the full automorphism group and $H$ is the subgroup quotienting to an elliptic curve when it exists}
    \label{fig: charge-3 curves automorphism groups.}
\end{figure}
\end{center}

\section{Classifying Curves by Automorphism Group}\label{sec: classifying curves by automorphism group}

In this section we determine the charge-3 monopole spectral curves we shall focus on, beginning with minimal restrictions and gradually imposing these.

A monopole spectral curve is a compact algebraic curve $\mathcal{C}$ lying in Euclidean mini-twistor space $\mathbb{MT}$, the space of oriented lines in Euclidean 3-space. If the direction of the oriented line is given by $\zeta $, an affine coordinate of $[\zeta_0:\zeta_1]\in \mathbb{P}^1$, and $\eta\in \mathbb{C}$ describes the point in the plane perpendicular to this through which the line passes then we have $\eta\partial_\zeta\in T\mathbb{P}^1 \cong \mathbb{MT}$. A generic charge-$k$ monopole spectral curve  may then be written as the zero set of a polynomial 
$
0 = \eta^k + \sum_{r=1}^{k} p_{2r}(\zeta_0,\zeta_1) \eta^{k-r},
$
where $p_{2r}$ is a homogeneous polynomial of degree $2r$ in $\zeta_0$, $\zeta_1$; equivalently a polynomial of degree $2r$ in $\zeta$.
Now $T\mathbb{P}^1$ is non-compact and two compactifications of this are common, either by inclusion in the (singular) weighted projective space 
$
\mathbb{P}\sp{1,1,2}=
\{(\zeta_0,\zeta_1, \eta)\in\mathbb{C}\sp3\setminus\{0\} \,|\, (\zeta_0,\zeta_1,\eta)\sim(\lambda\zeta_0,\lambda\zeta_1,\lambda^2\eta), \ \lambda\in\mathbb{C}\sp\ast \},
$
or (as by Hitchin) in the Hirzebruch surface $\mathbb{F}_2$. We adopt the former view and note that the singular point $[0:0:1]$ does not lie in  $T\mathbb{P}^1$ and hence on $\mathcal{C}$.
Next, via the Veronese embedding, we have $\iota:\mathbb{P}\sp{1,1,2}
\hookrightarrow\mathbb{P}\sp{3}$, $\iota([\zeta_0:\zeta_1: \eta])=
[\zeta_0^2:\zeta_0\zeta_1:\zeta_1^2: \eta]$. Under this a homogeneous polynomial of degree $2r$ becomes a homogeneous polynomial of degree $r$ in the new coordinates and $\mathbb{P}\sp{1,1,2}$ becomes a cone, a quadric, over the cone point $\iota([0:0:1])$ (see, for example, \cite[\S 8.2.11]{Vakil2010}). Thus a monopole spectral curve may be viewed as the complete intersection of a quadric cone and a degree-$k$ hypersurface in $\mathbb{P}\sp{3}$. This is known to be a curve of genus $(k-1)^2$ (see for example \cite[Exercise V.2.9]{Hartshorne1977}) which is non-hyperelliptic for $k\ge3$ (\cite[Exercise IV.5.1]{Hartshorne1977}). For $k=3$ we then have that $\mathcal{C}$ is a non-hyperelliptic genus-4 curve.

In 1895 Wiman \cite{Wiman1895} classified all non-hyperelliptic genus-4 curves by their automorphism group and gave explicit defining equations for these. Wiman's classification had two families: curves arose either as the intersections of a cubic surface and non-singular quadric in $\mathbb{P}\sp{3}$, or as the intersection of a cubic surface and quadric cone in $\mathbb{P}\sp{3}$. Thus charge-3 monopole spectral curves with automorphism group must lie in Wiman's second family. (The two rulings of the non-singular quadric of Wiman's first family lead to projections from the curve to $\mathbb{P}\sp{1}\times\mathbb{P}\sp{1}$, which is relevant for spectral curves of hyperbolic monopoles; this will be developed elsewhere.) We note that although $\Aut(\mathbb{P}^{3})=\PGL(4,\mathbb{C})$ differs from $\Aut(\mathbb{P}^{1, 1, 2})\cong \mathbb{C}\sp{3}\rtimes\left(\GL(2,\mathbb{C}))/\{\pm \id_2\} \right)$ (with the natural action of $\GL(2,\mathbb{C})$ on $(\zeta_0,\zeta_1)$ inducing that on $(\zeta_0^2,\zeta_0\zeta_1,\zeta_1^2)$), Wiman, in determining his normal forms, considered only those transformations of $\Aut(\mathbb{P}^{3})$ that preserved the cone, and so his normal forms include all possible charge-3 monopole spectral curves.
In Table \ref{tab: all curves} we give those curves in Wiman's classification which lie on a cone presenting\footnote{We set $y=1$ in Wiman's notation, so as to make clear the connection to monopole spectral curves.} these in terms of a curve given by the vanishing of a polynomial $P(x, z)$. We also write down their full automorphism group $G := \Aut(\mathcal{C})$ and the corresponding signature $c := c_G = (g_0; c_1, \dots, c_r)$ giving the quotient genus $g_0 = g(\mathcal{C}/G)$ and the ramification indices $c_i$ of the quotient map $\mathcal{C} \to \mathcal{C}/G$ (see \cite{Magaard2002}). These have been calculated with the help of the information available from \cite{LMFDB}. We make some remarks about Table \ref{tab: all curves}.
\begin{itemize}
    \item The label $D_n$ refers to the dihedral group of order $2n$, following the convention \cite{LMFDB}. 
    \item Wiman's parameters are to be understood as generic: there may exist specific values of the parameters for which the automorphism group is larger than that indicated.
    \item Wiman provides a form where the $z^2$ term is always zero, equivalent to centring the monopole. 
    \item All the curves given are irreducible, so we can only find reducible spectral curves as limiting members of the above families. 
    \item The completeness of the above data on signatures and elliptic quotients is reliant on the completeness of the data of the LMFDB. 
    \item We recognise the curve with $C_3 \times S_4$ symmetry as corresponding to the tetrahedrally-symmetric monopole. 
\end{itemize}

\begin{center}
    \begin{table}
        \centering
        \begin{tabular}{c|c|c|c|c}
            $P$  & $G$ & $c_G$ & $H$ & $c_H$ \\ \hline \hline 
        $z^3 + z(ax^4 + bx^2 + c) + (dx^6 + ex^4 + fx^2 + g)$ & $C_2$ & $(1; 2^6)$ & $C_2$ & $(1; 2^6)$ \\
        $z^3 + z(ax^{4} + b x^{2} + c)+ d x(x^{4} + e x^{2}  + f )$ & $C_2$ & $(2; 2^2)$ & & \\
        $z^3 + z[a(x^{4} + 1) + b x^{2} ] + x[c(x^{4} + 1) + d x^{2}] $  & $C_2 \times C_2 $& $(0; 2^7)$ & & \\
        $z^3 + z[a(x^{4} + 1) + b x^{2}]  + x(x^{4} -  1)$ & $C_2 \times C_2$ & $(1; 2^3)$ & $C_2^2$ & $(1; 2^3)$ \\
        $z^3 + a zx^2 + x(x^{4} +  1)$  & $D_4$ & $(0; 2^4, 4) $ & $C_4$ & $(1; 4^2)$ \\
        $z^3 + z(x^4 + a) + (b x^{4} + c)$  & $C_4$ & $(0; 2, 4^4) $ &  &  \\
        $z^3 + azx^{2} + x^{6} + b x^{3} + 1$  &$ S_3$ &$ (0; 2^6)$ & & \\
        $z^3 + a zx^2 + x^{6} + 1$ & $D_6$ & $(0; 2^5)$ & $S_3$ \text{and} $C_6$ & $(1; 2^2)$ \\
        $z^3 + z( a x^{3} + b ) + (x^{6} + c x^{3} + d )$  &$ C_3$ & $(1; 3^3)$ & $C_3$ & $(1; 3^3)$ \\
        $z^3 + a z(x^{3} + 1) + (x^{6} + 20 x^{3} - 8 )$ & $A_4$ & $(0; 2, 3^3)$ & & \\
        $z^3 + az + x^{6} + b $ & $C_6$ & $(0; 2, 6^3)$ & & \\ 
        $z^{3} + z + x^{6}$ & $C_{12}$ & $(0; 4, 6, 12)$ & & \\
        $z^3 + a z + x^{5} + b $ & $C_5$ & $(0; 5^4)$ & & \\ 
        $z^3 + z +x^{5}$ & $C_{10}$ & $(0; 5, 10^2)$  & & \\
        $z^3 - (x^6 + ax^5 + bx^4 + cx^3 + dx^2 + ex + f)$ & $C_3$ & $(0; 3^6)$ & & \\
        $z^3 -( x^{6} + ax^{4} + b x^{2} + 1)$ & $C_6$ & $(0; 2^2, 3^3)$ & & \\
        $z^3 -x(x^{4} +a x^{2} + 1)$ & $C_6 \times C_2$ & $ (0; 2^2, 3, 6)$ & & \\
        $z^3 -(x^{6} + ax^{3} + 1)$ & $C_3 \times S_3$ & $(0; 2^2, 3^2)$ & & \\
        $z^3-(x^{5} + 1) $& $C_{15}$ & $(0; 3, 5, 15)$ & & \\
        $z^3-(x^{6} + 1)$ & $C_6 \times S_3$ & $(0; 2, 6^2)$ & & \\
        $z^3-x(x^{4} + 1)$ & $C_3 \times S_4$ & $(0; 2, 3, 12)$ & $A_4$ & $(1;2)$ \\
        \end{tabular}
        \caption{Potential charge-3 monopole spectral curves with nontrivial automorphism group and those (with subgroups) quotienting to genus 1.}
        \label{tab: all curves}
    \end{table}
\end{center}

Not all curves on the list will yield  monopoles spectral curves; for example by the following result.

\begin{prop}[\cite{Braden2010a}]\label{c3tetra}
    There are only two curves in the family $\eta^3 + \chi(\zeta^6 + b \zeta^3 + 1)=0$, $\chi, b \in \mathbb{R}$, that correspond to BPS monopoles; these are tetrahedrally-symmetric monopole spectral curves.  
\end{prop}

\subsection{Moduli Space Dimension}\label{sec: moduli space}
We now also briefly discuss another aspect of this group theoretic approach that can be used to identify particularly tractable monopole curves.

There is a moduli space $N_k$ of charge-$k$ monopoles up to gauge transform, which one typically enlarges by a phase to $M_k$ which has a natural action of the Euclidean group $\operatorname{E}(3)$ and circle group $U(1)$. Associated to this is the submanifold of (strongly-) centred charge-$k$ monopoles $M_k^0 \subset M_k$ with action of the orthogonal group $\operatorname{O}(3)$ which parametrises monopoles up to gauge transform with fixed centre \cite{Atiyah1988, Manton2004}. The spectral curves corresponding to such monopoles have $a_1(\zeta)=0$, while the corresponding Nahm data has $\Tr(T_i)=0$. $M_k^0$ is a totally geodesic manifold of real dimension $4(k-1)$, and given $G \leq \operatorname{O}(3)$ we may consider the submanifold of the moduli space of $G$-invariant strongly-centred monopoles $(M_k^0)^G$, which is also totally geodesic \cite{Bielawski2020}. This we shall distinguish from $(M_k^0)^{[G]}$ which is the moduli space of monopoles invariant under an element of the conjugacy class of $G$ inside $\operatorname{SO}(3)$; clearly $(M_k^0)^G \subset (M_k^0)^{[G]}$ and when $G$ is discrete we have $\dim_\mathbb{R} (M_k^0)^G = \dim_\mathbb{R} (M_k^0)^{[G]} - \dim_\mathbb{R} \operatorname{SO}(3)$. As scattering of magnetic monopoles is approximated by geodesic motion in the moduli space \cite{Stuart1994}, if one can find $G \leq \operatorname{O}(3)$ for which $\dim (M_k^0)^G = 1$, it is known that the corresponding 1-parameter family of monopoles corresponds to a scattering process. We may use group theory to help identify such families via the following result. 

\begin{prop}
    For discrete $G$, $\dim_{\mathbb{R}} (M_3^0)^G \leq  3g_0 - 3 + r$.
\end{prop}
\begin{proof}
\cite[Lemma 3.1]{Magaard2002} gives that the complex dimension of each component of the locus of equivalence classes of genus $g \geq 2$ curves admitting an action of a group isomorphic to $G$ with signature $c$ is (provided it is non-empty) $\delta(g, G, c):= 3(g_0-1)+r = \dim_{\mathbb{C}} \mathcal{M}_{g_0, r}$ the moduli space of genus $g_0$ curves with $r$ marked points. The $\operatorname{SO}(3)$ action on $(M_3^0)^{[G]}$ is trivial on the moduli space of curves because it induces a birational isomorphism. The result then follows as the $\operatorname{SO}(3)$ orbits of the moduli space of monopoles will form a component of this locus, hence\footnote{As described in Earle, points in the moduli space of genus $g \geq 2$ Riemann surfaces correspond to equivalence classes of complex structures on a fixed smooth surface of the corresponding genus. There is a natural conjugate involution on this moduli space of complex structures, such that any complex structure fixed under an orientation-reversing involution of the surface corresponds to a fixed point of the conjugate involution of the moduli space. Hence we may consider the moduli space $\mathcal{M}_{g_0, r}^\tau \subset \mathcal{M}_{g_0, r} $ of curves invariant under $\tau$ by this procedure.} $\dim_\mathbb{R} (M_3^0)^G = \dim_\mathbb{R} (M_3^0)^{[G]} - \dim_\mathbb{R} \operatorname{SO}(3) \leq \dim_\mathbb{R} \pround{\mathcal{M}_{g_0, r}^\tau}$, and using the fact from Teichm\"uller theory that $\dim_\mathbb{R} \pround{\mathcal{M}_{g_0, r}^\tau} = \dim_\mathbb{C}\mathcal{M}_{g_0, r}$ \cite[\S3.1]{Earle1971}.
\end{proof}
\begin{remark}
    Note in the above we could have used $H \leq G$ and its corresponding signature, but this would have given a weaker bound as $\delta(g, G, c_G) \leq \delta(g, H, c_H)$ \cite{Magaard2002}. 
\end{remark}

The remarkable fact about this is that the bound depends on the signature only.  The curves of \cite{Hitchin1995}, with automorphism group $H$ of order $2k(k-1)$ for $k=3, 4, 6$,  have signature $c_H = (1;k-1)$ \cite[Proposition 4]{Hitchin1995}. For these curves and all known monopole spectral curves of charge $3$ we have $\dim (M_k^0)^H = \delta(g, H, c_H) - 1$, and hence one might conjecture that this is always true in the case $g_0=1$. A calculation for the case of inversion-symmetric monopoles considered in \cite{Bielawski2020} for which there is a $C_2$ action $\eta \mapsto -\eta$ shows that such a conjecture would certainly not be true for all $g_0$.

\subsection{Genus-1 Reductions}\label{sec: genus-1 reductions}
Table \ref{tab: all curves} gives us a list of putative spectral curves with symmetry before we have imposed the further constraints of Hitchin. We know from \cite{Hitchin1983} that Nahm's equations correspond to a linear flow in the Jacobian of the corresponding spectral curve $\mathcal{C}$; the direction of this linear flow is given by the Ercolani-Sinha vector ${\bm{U}}$ \cite{Ercolani1989}. Braden \cite{Braden2011} has shown that when we have a symmetry group $G$ we may be able to reduce to the quotient curve $\mathcal{C} \overset{\pi}{\to} \mathcal{C}' :=\mathcal{C}/G$ and reduced Ercolani-Sinha vector $\bm{U}'$ when  ${\bm{U}} = \pi^\ast \bm{U}'$. For example charge-$k$ monopoles with $C_k$ symmetry reduce to questions about a genus-$(k-1)$ hyperelliptic curve. The $k=3$ case was studied in \cite{Braden2011a}.

 Notwithstanding the attendant simplifications, the list of Table \ref{tab: all curves} is too long for the purposes of this letter and we require a further criterion to reduce this. Here we adopt the following: does the genus-$4$ spectral curve (assumed with real structure) quotient (either by $\Aut(\mathcal{C})$ or a subgroup) to an elliptic curve? The rationale for this is that the remaining of Hitchin's conditions are most straightforwardly answered for elliptic curves; equivalently the Ercolani-Sinha
 constraint becomes one on the real period of an elliptic curve. There are also a number of curves known with this property \cite{Hitchin1995, Houghton1996, Houghton1996b, Houghton1996c}. 
 
Thus we seek curves $\mathcal{C}$ with real structure from Wiman's list for which there exists $H \leq \Aut(\mathcal{C})$ such that $g(\mathcal{C}/H)=1$. Here we may use the database of \cite{LMFDB} which has enumerated all the possible $H$ and the corresponding signatures for genus-4 curves. We may then use our knowledge of the explicit forms of the curves to match up these cases, which leaves us with the reduced list in the final two columns of Table \ref{tab: all curves}. As previously noted, the $H = A_4$ case corresponds to the tetrahedrally-symmetric monopole \cite{Hitchin1995}, and the $H=C_4$ case has already been solved in \cite{Houghton1996c}. We also see that the cases $H= S_3$ and $H=C_6$ arise from the same curve, indicating that the curve has two distinct quotients to an elliptic curve. 

In the following sections we will investigate in more detail the two new cases $H=C_6$ (or equivalently $H=S_3$) with full automorphism group $G=D_6$, and $H = V_4$ (with full automorphism group $G=V_4$); we do not treat the $C_2$ case here. We will begin with the $D_6$ case which is both illustrative and simpler, though ultimately the new solutions and their scattering family are less interesting. 

Before turning to these however we may complete the proof of Theorem \ref{charge3auto}. With the exception of the $H=C_3$ curve, imposing reality on the curves with groups $H$ listed in Table \ref{tab: all curves} yields the curves of Theorem \ref{charge3auto} (and in the same order). Note not all $M \in \GL_2(\mathbb{C})/\pangle{-1} \leq \Aut(\mathbb{P}^{1, 1, 2})$ will commute with the action of $\tau$, but $S \in \operatorname{SU}(2)$ will. One may use Schur decomposition to write $M = S T S^{-1}$ for some $S \in \operatorname{SU}(2)$, $T$ an upper-triangular matrix, and so when imposing reality on Wiman's normal forms one should consider the orbits under upper-triangular matrices. The only real forms present in the orbit of the $G=H=C_3$ family have $a=b=0$ in the corresponding defining equation $P$; the resulting curve then lies in the family described by Proposition \ref{c3tetra}. Only the tetrahedrally-symmetric monopole within this family quotients to an elliptic curve and by a rotation this may be written as $\eta^3 + ia\zeta(\zeta^4 - 1)=0$, the final entry of the Theorem. We have thus established Theorem \ref{charge3auto}.

\section{\secmath{D_6} Monopoles}\label{D6 Monopoles}

To understand the spectral curves of this section it helps to first begin with the general centred $C_k$ invariant spectral curve $\mathcal{C}$ (with reality imposed),
\begin{equation}\label{cyclick}
\eta^k+ \alpha_2\eta\sp{k-2}\zeta^2+
 \alpha_3\eta\sp{k-3}\zeta^3+ \ldots+ \alpha_{k-1}\eta\,\zeta\sp{k-1}+ \alpha_k\zeta\sp{k}+\beta [
\zeta\sp{2k}+(-1)\sp{k}]=0,
\end{equation}
where $\alpha_k,\beta \in\mathbb{R}$. This is invariant under
$s: (\eta,\zeta)\rightarrow (\omega\eta,\omega\zeta)$,
$\omega=\exp(2\pi\imath/k)$ with $C_k=\pangle{s}$. The work of \cite{Braden2011} shows (\ref{cyclick}) is the unbranched cover of the hyperelliptic curve
\begin{align}\label{hyperck}
y^2&=(x^k+ \alpha_2 x\sp{k-2}+ \alpha_3 x\sp{k-3}+ \ldots+  \alpha_{k})^2 -(-1)^k
4\beta^2,
\end{align}
where $x=\eta/\zeta$ and $y=\beta[ \zeta\sp{k} - (-1)^k\zeta\sp{-k}]$. The curve (\ref{cyclick}) also has the symmetry 
 $t:(\zeta, \eta) \mapsto (-1/\zeta, -\eta/\zeta^2)$ and
 $G=\pangle{s,t}=D_k$ is the full automorphism group. The transformation $t$
 corresponds to a reflection $\mathbf{\underline{v}}\rightarrow
 \diag(1,-1,1)\mathbf{\underline{v}}$
 in $\operatorname{O}(3,\mathbb{R})$ and it becomes the
 hyperelliptic involution $t:(x,y)\rightarrow (x,-y)$ on the
 quotient curve. For $k=3$ we are describing the curve in Table 
\ref{tab: all curves} with full automorphism group $G=S_3$ ($=D_3$ in \cite{LMFDB} notation).
Further, the work of \cite{Braden2011} shows that Nahm's equations for every charge-$k$ monopole with $C_k$ rotational symmetry are equivalent to the $A_{k-1}\sp{(1)}$\footnote{This is the notation of \cite{Helgason1978}.} affine Toda equations
(in real Flaschka variables)
\begin{equation}\label{gentoda}
a_i^\prime = \frac{1}{2}a_i (b_i - b_{i+1}), \quad b_i^\prime = a_i^2 - a_{i-1}^2, 
\end{equation}
where $i$ is taken mod $k$, and we use$\phantom{.}^\prime$ to denote $\frac{d}{ds}$. 
These equations may also be found in other ways.
In Appendix \ref{sec: initial computation of nahm matrices} 
we describe how taking the $C_k$ invariant polynomials $Q_i = \zeta_0^i \zeta_1^i, i=1, \dots, k$ and $Q_{k+1} = \zeta_0^k - \zeta_1^k$
as the inputs to the procedure of \cite{Hitchin1995} we obtain the equations (\ref{gentoda}).

In general the solutions to the Toda system (\ref{gentoda}) linearise on the genus-$(k-1)$ Jacobian of the curve (\ref{hyperck}).
For $k=3$ this was the approach taken in \cite{Braden2011a} where a family of monopoles including the tetrahedrally-symmetric monopole was investigated. Are simplifications possible? In the remainder of this section we shall show that for $k=3$ a one-parameter family of elliptic Nahm data exists with $C_3$ symmetry. As this is contrary to results in the literature we begin with four results that lead to such an
elliptic reduction before determining the parameters that yield
Nahm data. Three particular points in the family will be identified before concluding with a description of the scattering described by the family.

\subsection{Four Lessons} For $k=3$ the equations (\ref{gentoda}) take the form (with $a_0\equiv a_3$, $b_0\equiv b_3$)
\begin{equation}
\begin{aligned}\label{todak3}
    a_0^\prime =& \frac{1}{2}a_0 ( b_3 - b_1), &
    a_1^\prime =& \frac{1}{2} a_1  (b_1 - b_2), &
    a_2^\prime =& \frac{1}{2} a_2 ( b_2 - b_3),\\
    b_1^\prime =& a_1^2 - a_0^2, &
    b_2^\prime =& a_2^2 - a_1^2 , &
    b_3^\prime =& a_0^2 - a_2^2 ,
\end{aligned}
\end{equation}
coming from Nahm matrices derived in Appendix \ref{sec: initial computation of nahm matrices} given by 
\begin{align}\label{nahmtoda}
    T_1 &= \frac{1}{2}\begin{pmatrix}
        0 & a_1 & -a_0 \\ -a_1 & 0 & a_2 \\ a_0 & -a_2 & 0  
    \end{pmatrix}, &
    T_2 &= \frac{1}{2i} \begin{pmatrix}
        0 & a_1 & a_0 \\ a_1 & 0 & a_2 \\ a_0 & a_2 & 0  
    \end{pmatrix}, &
    T_3 &= \frac{-i}{2} \begin{pmatrix} b_1 & 0 & 0 \\ 0 & b_2 & 0 \\ 0 & 0 & b_3 \end{pmatrix}.
\end{align}
Here we find the constants
\begin{align*}
    \alpha_2 &= b_{1} b_{2} + b_{1} b_{3} + b_{2} b_{3} + a_{0}^{2} + a_{1}^{2} + a_{2}^{2}, \quad
    \alpha_3 = b_{1} b_{2} b_{3} + b_{1} a_{2}^{2} + b_{2} a_{0}^{2} + b_{3} a_{1}^{2}, \quad
    \beta = a_0 a_1 a_2,
\end{align*}
for the (centred) spectral curve
\begin{align*}
\eta^3 + \alpha_2 \eta \zeta^2 + \alpha_3 \zeta^3 + \beta (\zeta^6 - 1)=0
\end{align*}
which covers
\begin{equation}\label{cyclick3}
y^2=(x^3+ \alpha_2 x+ \alpha_3)^2 + 4\beta^2.
\end{equation}
We have $6$ differential equations, $6$ variables and three conserved quantities.

\subsubsection{Direct Simplification} Appendix \ref{sec: initial computation of nahm matrices} shows that we may use
the constants $\alpha_2,\alpha_3, 0 = \sum b_i$ to eliminate the $b_i$, resulting in the equations
\begin{align*}
0 &= \sum_{i=0}^2 a_i^2 - \alpha_2 - \frac{1}{3}(d_1^2 + d_1 d_2 + d_2^2), \ 
0 = a_1^2 d_2 - a_2^2 d_1 + \alpha_3 + \frac{1}{3}\alpha_2(d_1 - d_2) + \frac{1}{27}(d_1 - d_2)^3,
\end{align*}
where we have introduced $d_i = {2a_i^\prime}/{a_i}$. 
Using  $\beta = a_0 a_1 a_2$ to eliminate $a_0$ we then have
two non-linear ODE's in two variables,
the maximal reduction one can achieve with generic $\alpha_i$ and $\beta$. The appendix shows further that if $a_1^2 = a_2^2$ additional simplification is possible; and that we may 
consistently set $a_1^2 - a_2^2 = 0$ provided $b_2 a_1^2 = 0$. As $b_2^\prime = a_2^2 - a_1^2$, this means we can consistently set $a_1^2 = a_2^2$ and $b_2 = 0$. Making these restrictions we find $\alpha_3=0$ and that we reduce to one equation
\begin{equation*}
a_1^2\left({2\frac{da_1}{ds}}\right)^2 = \beta^2 + 2 a_1^6 - \alpha_2 a_1^4. 
\end{equation*}
Upon setting $u=a_1^2$ this becomes
\begin{equation}\label{redeqn}
\left({\frac{du}{ds}}\right)^2 = \beta^2 + 2 u^3 - \alpha_2 u^2 , 
\end{equation}
to which we shall return. We record that the $j$-invariant of the
associated elliptic curve $y^2=\beta^2 + 2 u^3 - \alpha_2 u^2$
is $16\alpha_2^6/(\beta^2[\alpha_2^3-27 \beta^2] )$.

\subsubsection{Sutcliffe's ansatz} Some time ago, in the context of Seiberg-Witten theory, Sutcliffe \cite{Sutcliffe1996a} gave an ansatz for charge-$k$ cyclically symmetric monopoles in terms of affine Toda theory. With 
$a_j=\gamma e^{(q_j-q_{j+1})/2}$, $b_j=q_j'$ equations (\ref{gentoda})
follow from a Hamiltonian\footnote{This Hamiltonian is unbounded from below: such is necessary as the monopole boundary
conditions require a pole at $s=0,2$ hence the momenta and correspondingly the potential must also be unbounded below. Thus while the dynamical system being described is integrable, a corresponding interpretation in terms of a mechanical system is less helpful. Further, while the $a_i$ will always be real 
throughout we have freedom to choose their sign and we 
will make sign choices for the $a_i$ where we cannot take $\log a_i$ and obtain real values for the associated Toda position variables. 
}
$
H=\frac12 \sum_{j=1}\sp{k} b_j^2-\sum_{j=0}\sp{k-1} a_j^2
=\frac12 \sum_{j=1}\sp{k} p_j^2-\gamma^2\sum_{j=0}\sp{k-1} e^{q_j-q_{j+1}}.
$
(The constant $\gamma$ here is to account for the constant $\prod_{i=0}\sp{k-1}a_i=\gamma^k=(-1)\sp{k-1}\beta$.)
Sutcliffe showed that for $k=2$ Nahm data could be constructed, but for $k=3$ although he could solve the equations he couldn't find solutions with the correct pole behaviour. The solution was obtained from the infinite chain solution as follows. We have
from
$$\left(\ln a_j^2\right)''=-a_{j-1}^2+2 a_{j}^2-a_{j+1}^2$$
and the standard elliptic function identity for the Weierstrass $\wp$-function
$$
\frac{d^2}{du^2}\,\ln[ \wp(u)-\wp(v) ]= -\wp(u+v) +2 \wp(u)-\wp(u-v)
$$
{that with $u= j u_0+t+t_0$ and $v=u_0$ then}
$$
\frac{d^2}{dt^2}\,\ln[  \wp(j u_0+t+t_0)-\wp(u_0) ]=- \wp([j+1] u_0+t+t_0) +2 \wp(j u_0+t+t_0)
 -\wp([j-1] u_0+t+t_0)
$$
{and we may identify}
$a_j=\wp(j u_0+t+t_0)- \wp(u_0)$. This yields the solution for the infinite chain and we must still impose periodicity to obtain a solution. Imposing periodicity yields (for $k=3$) that
$a_j=\wp(2j K/3 +t)-\wp(2K/3)$ which is equivalent to the solution of \cite{Sutcliffe1996a} which is given in Jacobi elliptic functions.\footnote{
To make connection with \cite{Sutcliffe1996a} we use Lawden's notation \cite[\S6.3.1, \S6.9]{Lawden1989}. Thus for $k=3$ we take $u_0=2K/3$. Now
\begin{align*}
dc^2(u)&=\frac{ \wp(u)-e_2}{\wp(u)-e_1} = 1 + \frac{ e_1-e_2}{\wp(u)-e_1} 
=1 + \frac1{ e_1-e_3}\left[{\wp(u+\omega_1)-e_1} \right]=
1+\left[{\wp(u+\omega_1)-e_1} \right],
\\
cs^2(2K/3)&=\wp(2K/3)-e_1.
\end{align*}
Note Sutcliffe's \lq$q_j^2$\rq\ is our $a_j$. Then \cite[3.41]{Sutcliffe1996a}  is
$dc^2(u)-1-cs^2(2K/3)=\wp(u+\omega_1)-\wp(2K/3)=\wp(u+K)-\wp(2K/3)$
and so with $u=2j K/3 +t+K$ (his choice of $\delta$) we get
$a_j=\wp(2j K/3 +t)-\wp(2K/3)$
and the corresponding asymptotics given in \cite[3.42-45]{Sutcliffe1996a}.} Now the ansatz employed here forces only one of the $a_j$ to be singular at any point, and this means the pole condition on the Nahm matrices cannot be satisfied.
If we are to find an alternative solution that does indeed yield a monopole then this would suggest that one appropriate route would be to pick a simplification which forces multiple variables to have poles simultaneously. Such is the case
when $a_1^2=a_2^2$ found previously.

\subsubsection{Imposing Symmetry on Nahm Matrices}
We next show that $a_1^2=a_2^2$ follows from the symmetry\footnote{We have the correspondences between the $\mathbb{R}\sp3$ transformations and $(\zeta,\eta)$ actions
$$ 
   r:= \begin{pmatrix}
        1 & 0 & 0 \\ 0 & -1 & 0 \\ 0 & 0 & -1  
    \end{pmatrix} \leftrightarrow
   (\zeta, \eta) \mapsto (1/\zeta, -\eta/\zeta^2),\quad
   rt:=\begin{pmatrix}
        1 & 0 & 0 \\ 0 & 1 & 0 \\ 0 & 0 & -1
    \end{pmatrix}\leftrightarrow
   (\zeta, \eta) \mapsto (1/\bar{\zeta}, -\bar{\eta}/\bar{\zeta}^2).
$$
If the rotation $s:=(\eta,\zeta)\rightarrow (\omega\eta,\omega\zeta)$ takes place in the $xy$-plane then  we obtain the point groups $D_{k}=\pangle{s,r}$, $C_{kv} \cong D_{k}=\pangle{s,t}$ and $C_{kh} \cong C_k \times C_2=\pangle{s,rt}$. The prismatic dihedral group $D_{kh}=\pangle{s,r,t}$ is obtained by adding any two of the above to the rotations, so giving the third. Abstractly $D_{kh}\cong D_k\times C_2$. }   $\mathbf{\underline{v}}\rightarrow A\mathbf{\underline{v}}:= \diag(1,-1,-1)\mathbf{\underline{v}}$ in $\operatorname{SO}(3,\mathbb{R})$. From \cite[Equation (8.169)]{Manton2004} we know that the conditions for the Nahm matrices to be symmetric under $A$ are that 
$$
T_1 = CT_1 C^{-1}, \quad -T_2 = C T_2 C^{-1}, \quad -T_3 = C T_3 C^{-1},
$$
for some constant invertible matrix $C$. Recalling the form of the $T_i$ from (\ref{nahmtoda}) we see that as $T_3$ is diagonal and traceless the only way to achieve the invariance under $A$ is if at least one of the $b_i$ is 0 and $C$ permutes the other two. By conjugating with a permutation matrix we can without loss of generality pick $b_2 = 0$ so $b_1 = -b_3$ which gives that the generic $C$ is $C = \begin{psmallmatrix}
    0 & 0 & a \\ 0 & b & 0 \\ c & 0 & 0
\end{psmallmatrix}$. Picking a generic $a, b, c$ we get 
\[
a_0(a+c) = a_1 a - a_2 b = a_1 c + a_2 b = a_1 b + a_2 c = a_2 a - a_1 b = 0.
\]
To avoid having an $a_i = 0$ we required $a= -c$, and so these reduce to 
\[
a_1 a - a_2 b = 0 = a_1 b - a_2 a ,
\]
and consequently $(a/b)^2=1$ and $a_1 =\pm a_2$ yielding the desired $a_1^2=a_2^2$. Note that this also means $\alpha_3 = 0$.

\subsubsection{Reduction of the spectral curve; folding.}
In order to get the curve with $D_6$ symmetry of Table \ref{tab: all curves} we must set $\alpha_3=0$. We have seen that for $k=3$ this is a consequence of the symmetry
 $r:(\zeta,\eta)\rightarrow (1/\zeta,-\eta/\zeta^2 )$. 
For general $k$ this means we keep only the even terms of (\ref{cyclick}),
\begin{equation}\label{dihk}
\eta^k+ \alpha_2\eta\sp{k-2}\zeta^2+
 \alpha_4\eta\sp{k-4}\zeta^4+ \ldots +\beta [
\zeta\sp{2k}+(-1)\sp{k}]=0.
\end{equation}
The full automorphism group of this curve is $D_{k}\times C_2$; 
for $k=3$ this is the curve with full automorphism group
$D_6\cong D_3\times C_2$ that we are interested in. 
Setting $x=\eta/\zeta$ in (\ref{dihk}) we have
\begin{align*}
x^k+ \alpha_2 x\sp{k-2}+ \alpha_4 x\sp{k-4}+ \ldots+  \alpha_k+\beta [
\zeta\sp{k}+\zeta\sp{-k}]&=0, &&k\ \text{even},\\
x^k+ \alpha_2 x\sp{k-2}+ \alpha_4 x\sp{k-4}+ \ldots+  \alpha_{k-1} x+\beta [
\zeta\sp{k}-\zeta\sp{-k}]&=0, &&k\ \text{odd}.
\end{align*}
If $y=\beta[ \zeta\sp{k} - (-1)^k\zeta\sp{-k}]$ then 
$r:(x,y)\rightarrow (-x, (-1)^{k-1}y)$; thus
$y$ is invariant under $r$ only for $k$ odd, in which case it will be a function on the quotient curve $\hat{\mathcal{C}}/\pangle{s,r}$; for $k$-even $v=xy$ is invariant. Thus we have curves
\begin{align*}
v^2&=x^2(x^k+ \alpha_2 x\sp{k-2}+ \alpha_4 x\sp{k-4}+ \ldots+  \alpha_k)^2
-4\beta^2 x^2&&k\ \text{even},\\ 
y^2&=(x^k+ \alpha_2 x\sp{k-2}+ \alpha_4 x\sp{k-4}+ \ldots+  \alpha_{k-1} x)^2
+4\beta^2 &&k\ \text{odd}.
\end{align*}
Setting $k=2l$ or $k=2l-1$ for the even and odd cases of the curves then with
$u=x^2$ we have these curves covering $2:1$ the curves
\begin{align}\label{diheven2}
v^2&=u(u^l+ \alpha_2 u\sp{l-1}+ \alpha_4 u\sp{l-2}+ \ldots+  \alpha_k)^2
-4\beta^2u &&k\ \text{even},\\ \label{dihodd2}
y^2&=u(u^{l-1}+ \alpha_2 u\sp{l-2}+ \alpha_4 u\sp{l-3}+ \ldots+  \alpha_{k-1})^2
+4\beta^2 &&k\ \text{odd}.
\end{align}
The first has genus $l$ and the second has genus $l-1$. 
Under the cyclic transformation, it was shown in \cite{Braden2011}
that
$$\frac{\eta^{k-2}d\zeta}{\partial_\eta P}=\pi\sp\ast\left(-\frac1k \frac{x^{k-2}dx}{y}\right)
$$
for the curve (\ref{hyperck})
and we observe that this differential is invariant under $r$ for $k$ both even and odd. Further
$$
\frac{x^{k-2}dx}{y}=\begin{cases} \dfrac{x^{2l-2}dx}{y}=\dfrac{x^{2l-2}du}{2xy}=
\dfrac{u^{l-1}du}{2v},\\
\dfrac{x^{2l-3}dx}{y}=\dfrac{x^{2l-4}du}{2y}=\dfrac{u^{l-2}du}{2y}.
\end{cases}
$$
In each case we obtain the maximum degree in $u$ differential on the corresponding hyperelliptic curve and the work of
\cite{Braden2011} tells us the Ercolani-Sinha vector, if it exists, will reduce to one on the quotient curve.

In particular the $k=3$ curve $y^2=(x^3+ \alpha_2 x)^2 + 4\beta^2$ covers the elliptic curve $\mathcal{E}=\mathcal{C}/H$,
$$y^2 =u(u+\alpha_2)^2+4\beta^2,$$
with $H=\pangle{s,r}\cong S_3$.
The
$j$-invariant of this curve 
is $j_\mathcal{E}=16\alpha_2^6/(\beta^2[\alpha_2^3-27 \beta^2] )$, the value observed earlier. We note that the genus-$2$ curve also covers the  elliptic curve $\mathcal{E}'=\mathcal{C}/H'$,
$$w^2 =u^2(u+\alpha_2)^2+4\beta^2 u,$$
where now $H'=\pangle{s,rt}\cong C_6$ with $w=xy$ the invariant coordinate. Because $\pi\sp\ast(du/(2w))=dx/y$ does not pull back to the differential appearing in the Ercolani-Sinha constraint we cannot solve the Hitchin constraints in terms of $\mathcal{E}^\prime$. We
record that the curve is in general distinct
$j_{\mathcal{E}'}=\psquare{\alpha_2^3 ( \alpha_2^3 - 24 \beta^2)^3}/\psquare{\beta^6 ( \alpha_2^3 - 27 \beta^2)}$. We have that ${\mathcal{E}}$ and
${\mathcal{E}'}$ are the two quotients identified in
Table \ref{tab: all curves}.

We remark that the reduction of the spectral curve we have just described may be understood directly in terms of the Toda equations and \lq folding\rq. For the $k=3$ case at hand set
$e\sp{\rho_i}:=a_i^2=\beta\sp{2/3}e\sp{q_i-q_{i+1}}$ (so that
$a_0a_1a_2=\beta$) and again take $b_j=q_j'$ and Hamiltonian 
$
H=\frac12 \sum_{j=1}\sp{k} b_j^2-\sum_{j=0}\sp{k-1} a_j^2
=\frac12 \sum_{j=1}\sp{k} p_j^2-\beta^{2/3}\sum_{j=0}\sp{k-1} e^{q_j-q_{j+1}}.
$
Then the Toda equations take the form
\begin{align*}
     \rho_i^{\prime \prime} &= 2e^{\rho_i} - e^{\rho_{i-1}} - e^{\rho_{i+1}} =\overline{K}_{ij} e\sp{\rho_j}
    \end{align*}
where $\overline{K}_{ij}$ is the extended Cartan matrix of $A_2$. Folding \cite{Olive1983} corresponds to the action $\rho_i\rightarrow\rho_{\sigma(i)}$ by a diagram automorphism $\sigma$ of
the extended Dynkin diagram: this retains integrability and here
corresponds to identifying $\rho_1=\rho_2:=\rho_{12}$, equivalently $a_1^2=a_2^2$. Using $e\sp{\rho_{0}}=\beta^2 e\sp{-2\rho_{12}}$ the equations of motion 
$\rho_{12}''=e\sp{\rho_{12}}-e\sp{\rho_{0}}$ and
$\rho_0''=2(e\sp{\rho_{0}}-e\sp{\rho_{12}})$ reduce to the one equation,
$$\rho_{12}''=e\sp{\rho_{12}}-\beta^2 e\sp{-2\rho_{12}}, $$
the ODE reduction of the Bullough-Dodd equation, a known integrable equation. This may be directly integrated. With
$u=e\sp{\rho_{12}}$ we obtain precisely (\ref{redeqn}). 
More generally we are seeing the reduction by folding $A_{2l-1}\sp{(1)}\rightarrow C_l\sp{(1)}$ for $k=2l$ even, and $A_{2[l-1]}^{(1)} \to A_{2[l-1]}^{(2)}$ for $k=2l-1$ odd, both coming from an order-2 symmetry of the Dynkin diagram.


\subsection{Solving for Nahm Data}
A number of different arguments lead us then to an elliptic reduction of the Toda equations for $k=3$ with corresponding ODE (\ref{redeqn}). The aim of this subsection is to show that from this ODE Nahm data can be constructed. In doing so we will use properties of hypergeometric functions, and we lay out some of these details in Appendix \ref{sec: hypergeometric function proofs}.

We have seen that the reduction leads to $a_1^2=a_2^2$ and $b_2=0$. In continuing to solve for the Nahm data one finds that the choice of sign of $a_2$ relative to $a_1$ does not affect the ability to impose the Hitchin constraints. Indeed, changing the choice of sign merely corresponds to changing the sign of $\beta$, and again as we will see this does not restrict the spectral curve. As such we take $a_2 = - a_1$ in what follows. Now setting  $\tilde{u} = u - \frac{\alpha_2}{6}$ and $\tilde s= s/\sqrt{2}$ we may transform
(\ref{redeqn}) into standard Weierstrass form with solution
\[
\tilde{u} = \wp((s - s_0)/\sqrt{2}; g_2, g_3),
\]
where $g_2 = \frac{\alpha_2^2}{3}$ and $g_3 = \frac{\alpha_2^3}{27} - 2 \beta^2$. Here we assume
$\Delta:=g_2^3 -27 g_3^2=4\beta^2(\alpha_2^3-27 \beta^2)\neq 0$  to avoid nonsingularity, commenting on the singular limits at the appropriate junctures. The
 $j$-invariant of the elliptic curve is as we have already seen
\[
j = 1728 \frac{g_2^3}{g_2^3 - 27 g_3^2} = \frac{16 \alpha_2^6}{\beta^2 (\alpha_2^3 - 27 \beta^2)} . 
\]

To be Nahm data we require that the Nahm matrices have a pole at $s=0$ which can be achieved by setting $s_0 = 0$. We can then express all the Flaschka variables as
\begin{align}\label{eqn: Flaschka variables in terms of wp}
    a_1 &= \pm\sqrt{\wp(s/\sqrt{2}; g_2, g_3) + \frac{\alpha_2}{6}}, &
    a_2 &= - a_1, &
    a_0 &= \frac{\beta}{a_1 a_2}, \\
    b_1 &= \pm\sqrt{2 a_1^2 + a_0^2 - \alpha_2}, &
    b_2 &= 0, &
    b_3 &= -b_1. 
\end{align}
We have some signs of the square roots to set above.
\begin{enumerate}[(i)]
    \item Using that, around $s=0$, $\wp(s/\sqrt{2}; g_2, g_3) \sim 2s^{-2} \Rightarrow a_1^2 \sim \frac{2}{s^2}$, we have $a_0 \sim \frac{\beta s^2}{2}$. The ODE for $a_0^\prime$, with $b_3 = -b_1$, gives 
    \[
    b_1 = - \frac{a_0^\prime}{a_0} \sim -\frac{(\beta s)}{(\beta s^2/2)} = -\frac{2}{s}.
    \]
    This requires us to take the negative square root for $b_1$ around $s=0$. We will want residues at $s=2$, and it will turn out by applying similar analysis that we need the positive root around $s=2$. These swap over when $b_1=0$, which corresponds to $a_1^\prime = 0$. As we see later this must happen at $s=1$. Alternatively one can see this from the observation that $a_1$ is even about $s=1$ by a judicious choice of period, and so $b_1 = \frac{2 a_1^\prime}{a_1}$ is odd about the same point.
    \item The sign of $a_1$ is a free choice, and does not affect the geometry of the monopole, hence in what follows below we always take the positive sign. 
\end{enumerate}

The corresponding Nahm matrices (\ref{nahmtoda})
have residues at $s=0$ given by 
\begin{align*}
    R_1 =\frac1{\sqrt{2}} \left(\begin{array}{rrr}
0 & 1 & 0 \\
-1 & 0 & -1 \\
0 & 1 & 0
\end{array}\right) , \quad
R_2 = \frac{i}{\sqrt{2}} \left(\begin{array}{rrr}
0 & -1 & 0 \\
-1 & 0 & 1 \\
0 & 1 & 0
\end{array}\right) , \quad
R_3 =i \left(\begin{array}{rrr}
1 & 0 & 0 \\
0 & 0 & 0 \\
0 & 0 & -1
\end{array}\right)
\end{align*}
which yield a 3-dimensional irreducible representation of $\mathfrak{su}(2)$. 

Next we require a simple pole at $s=2$ again forming a
3-dimensional irreducible representation.
There are two ways to achieve a residue at $s=2$:
\begin{enumerate}[(i)]
    \item have that $2/\sqrt{2} = \sqrt{2}$ is in the lattice corresponding to the values $g_2, g_3$, or
    \item have that around $s=2$, $\wp(s/\sqrt{2}; g_2, g_3) \sim -\frac{\alpha_2}{6} + \mathcal{O}(s-2)$.
\end{enumerate}
These correspond to having $a_1$ and $a_0$ be singular at $s=2$ respectively. (Because of the constant $\beta$ they cannot both be singular.) One can check that the second condition would give a reducible representation at $s=2$ (as again only one of the $a_i$ have a pole here) and so we discount it.

Focusing then on the first condition, one way to fix the real period of the associated lattice is to invert the $j$-invariant of the elliptic curve corresponding to $g_2, g_3$ to give the period $\tau$. Here this is most readily achieved by solving the quadratic (for example, see \cite{Berndt1995})
\begin{align*}
4\alpha (1-\alpha) &= \frac{1728}{j} = 108 (\beta^2/\alpha_2^3) \psquare{1 - 27(\beta^2/\alpha_2^3)} ,
\end{align*}
for which we see the two solutions are $\alpha = \frac{27 \beta^2}{\alpha_2^3}$, and $1-\alpha$. The corresponding normalised period is 
\[
\tau = \tau(\alpha) := i \frac{{}_2 F_1 (1/6 ,5/6, 1 ; 1-\alpha)}{{}_2 F_1 (1/6 ,5/6, 1 ; \alpha)}. 
\]
Some analytic properties of this function we need are given in Appendix \ref{proptau}.
\begin{remark}
    If we had taken the other root $\alpha$ in the numerator of the hypergeometric function then this would give the period $-1/\tau$. 
\end{remark}
As we want the lattice corresponding to $g_2, g_3$ to be $\sqrt{2}\mathbb{Z} + \sqrt{2}\tau \mathbb{Z}$, we get the transcendental equations 
\begin{align*}
    \frac{1}{3}\alpha_2^2 = \frac{1}{4} g_2(1, \tau), \quad
    \frac{1}{27} \alpha_2^3 - 2 \beta^2 = \frac{1}{8} g_3 \pround{1, \tau}.
\end{align*}
For any given value of $\alpha \in (0,1)$, let $\alpha_2^2 = \frac{3}{4} g_2(1, \tau)$. We then have two equations defining 
$\beta$:
\begin{align*}
    \beta^2 = \frac{\alpha \alpha_2^3}{27}, \quad \beta^2 = \frac{1}{2} \psquare{\frac{1}{27} \alpha_2^3 - \frac{1}{8} g_3(1, \tau)}.
\end{align*}
To have a valid solution we must have that the two equations
are consistent with each other, which one can check (see Appendix \ref{consistency}) is equivalent to $\operatorname{sgn}(g_3(1, \tau)) = \operatorname{sgn}(\alpha_2)\operatorname{sgn}(1-2\alpha)$. 
A consideration of the information given about $\tau$ and $g_3$ in Appendix \ref{sec: hypergeometric function proofs} tells us that we only get solutions in the region $\alpha \in [0,1]$, where $\alpha = 0, 1$ really correspond to the limits $\lim_{\epsilon \to 0^+} \epsilon, 1-\epsilon$ respectively. 

In order to exclude the possibility of other poles of the Nahm matrices in the region $s \in (0, 2)$, it is necessary that for all $s \in (0,2)$
\[
 \wp(s/\sqrt{2}; g_2, g_3) + \frac{\alpha_2}{6} > 0.
\]
We know that (i) $\wp$ takes its minimum at $s=1$; (ii) that the minimum value is the most-positive root of the corresponding cubic $4 \wp^3 - g_2 \wp - g_3 = 0$; (iii)
that this root is positive \cite[\S 23.5]{dlmf2022}. Therefore  
there are no other poles in $(0,2)$. Further as $\alpha_2\ne0$ has the same sign as $\alpha$, then $\alpha_2 > 0$ so $a_1^2 > 0$. Therefore we know that $a_1$ is always real, and hence so are all the Flaschka variables, thus giving all the Nahm variables being real as desired. 

The remaining condition required for valid Nahm data is that $T_i(s) = T_i(2-s)^T$. The nature of the Weierstrass $\wp$ is such that $\wp((2-s)/\sqrt{2}; g_2, g_2) = \wp(s/\sqrt{2}; g_2, g_3)$, so we automatically have that $a_1(2-s) = a_1(s)$, $a_0(s) = a_0(2-s)$. Moreover, because of the change in the sign of the square root giving $b_1$ at $s=1$, we have that $b_1(2-s) = - b_1(s)$. Taken together these ensure the desired symmetry of the Nahm matrices and we have a one-parameter family of
new solutions.

As such, we have now proven the following theorem.
\begin{theorem}\label{thm: D6 monopole spectral curves}
Given $\alpha \in [0,1]$, define 
        \[
        \tau = \tau(\alpha) = i \frac{{}_2 F_1 (1/6 ,5/6, 1 ; 1-\alpha)}{{}_2 F_1 (1/6 ,5/6, 1 ; \alpha)}. 
        \]
        Solving 
        \begin{align*}
    \frac{1}{3}\alpha_2^2 = \frac{1}{4} g_2(1, \tau), \quad \frac{1}{27} \alpha_2^3 - 2 \beta^2 = \frac{1}{8} g_3 \pround{1, \tau},
\end{align*}
with $\operatorname{sgn}(\alpha_2) = \operatorname{sgn}(\alpha)$ yields a monopole spectral curve with $D_{6}$ symmetry
\[
\eta^3 + \alpha_2 \eta \zeta^2 + \beta(\zeta^6 - 1) = 0. 
\]
Moreover, the Nahm data is given explicitly in terms of $\wp$-functions by (\ref{nahmtoda}) and (\ref{eqn: Flaschka variables in terms of wp}).
\end{theorem}

\subsection{Distinguished Curves}

Having solved for general $\alpha\in [0,1]$ we now investigate the special values of $\alpha = 0, 1/2, 1$. 

\subsubsection{\secmath{\alpha = 0^+}}

The limit $\alpha \to 0$ corresponds to $\tau \to +i \infty$, and we have using the asymptotic expansion of the Eisenstein series that $g_2(1, \tau) \to \frac{4\pi^4}{3}$, $g_3(1, \tau) \to \frac{8 \pi^6}{27}$, so $\alpha=0$ is indeed a solution with $\beta=0$, $\alpha_2 = \pi^2$. This recreates the well known axially-symmetric monopole with spectral curve $\eta(\eta^2 + \pi^2\zeta^2) = 0$ \cite{Hitchin1982, Hitchin1983}.

If we had $\beta=0$ from the beginning (and so $\Delta=0$, and for $\alpha_2\ne0$ then $\alpha=0$),  we would have found a singular elliptic curve 
\[
4 \tilde{u}^3 - \frac{1}{3}\alpha_2^2 \tilde{u} - \frac{1}{27}\alpha_2^3 = 4 \pround{\tilde{u} + \frac{\alpha_2}{6}}^2 \pround{\tilde{u} - \frac{\alpha_2}{3}} , 
\]
with solution to the corresponding ODE (using known integrals) 
given by
\begin{align*}
\tilde{u} &= \frac{\alpha_2}{3} + \frac{\alpha_2}{2} \tan^2 \psquare{\frac{\sqrt{\alpha_2}}{2}(s - s_0)}, \quad
a_1 = \sqrt{\frac{\alpha_2}{2}} \sec \psquare{\frac{\sqrt{\alpha_2}}{2}(s - s_0)}.
\end{align*}
We could then manufacture the right residue at $s=0$ by having $s_0 = \frac{\pi}{2} \cdot \frac{2}{\sqrt{\alpha_2}}$. To get the correct periodicity, we would require that
$
\frac{\pi}{2} = \frac{\sqrt{\alpha_2}}{2} (2 - s_0)$ and consequently that $\alpha_2 = \pi^2$ again giving the
axially-symmetric monopole.

\subsubsection{\secmath{\alpha = 1^-}}

To get this limit, we use $\tau(1^-) = -1/\tau(0^+)$, so 
\[
g_2(1, \tau(1^-)) = g_2(1, -1/\tau(0^+)) = \tau(0^+)^4 g_2(1, \tau(0^+)) = \frac{1}{\tau(1^-)^4} \frac{4 \pi^4}{3},
\]
and likewise for $g_3$. Solving gives 
\[
\alpha_2 \sim -\pround{\frac{\pi}{\tau}}^2, \quad \beta \sim \pm \frac{i}{3 \sqrt{3}}\pround{\frac{\pi}{\tau}}^3,
\]
or equivalently writing $\tau = i \epsilon$ for $0 < \epsilon \ll 1$,
\[
\alpha_2 \sim 3\pround{\frac{\pi}{\sqrt{3}\epsilon}}^2, \quad \beta \sim \pm \pround{\frac{\pi}{\sqrt{3}\epsilon}}^3,
\]
The corresponding spectral curve thus factorises as 
\begin{align*}
0 &= \eta^3 + 3\pround{\frac{\mp \pi}{\sqrt{3}\epsilon}}^2\eta\zeta^3 - \pround{\frac{\mp \pi}{\sqrt{3}\epsilon}}^3(\zeta^6 - 1) , \\
&= \psquare{\eta - \pround{\frac{\mp \pi}{\sqrt{3}\epsilon}}(\zeta^2 - 1)} \psquare{\eta - \pround{\frac{\mp \pi}{\sqrt{3}\epsilon}}(\omega\zeta^2 - \omega^2)}\psquare{\eta - \pround{\frac{\mp \pi}{\sqrt{3}\epsilon}}(\omega^2\zeta^2 - \omega)}  .
\end{align*}
This corresponds to three well-separated 1-monopoles on the vertices of an equilateral triangle in the $x,y$-plane with side length $\frac{\pi}{\epsilon}$ \cite{Hitchin1995}. As $\epsilon$
tends to zero these three vertices tend to the point $\infty$, the singular degeneration to the cuspidal elliptic curve with $\Delta=0$ and $\alpha=1$. 

\subsubsection{\secmath{\alpha = 1/2}}
In this case $\tau=i$, and the lattice is the square lattice. The values of $g_2, g_3$ for this lattice are known explicitly \cite[23.5.8]{dlmf2022}, giving the equations
\begin{align*}
    \frac{1}{3}\alpha_2^2 = \frac{1}{4} \frac{\Gamma(1/4)^8}{16 \pi^2} , \quad
    \frac{1}{27} \alpha_2^3 - 2 \beta^2 = 0 \Rightarrow \alpha_2 = \frac{\sqrt{3}\Gamma(1/4)^4}{8 \pi}, \quad \beta = \pm \frac{\Gamma(1/4)^6}{32(\sqrt{3}\pi)^{3/2}}.
\end{align*}

\begin{remark}
    The coefficients seen here are the same, up to a sign, as those of a distinguished monopole found in \cite{Houghton1996c}. This is no accident, but arises because the square lattice is behind the distinguished ``twisted figure-of-eight" monopole, as we show later in \S\ref{sec: D4 monopoles}.
\end{remark}

\subsection{Scattering}

To complete our understanding of these monopoles we discuss the corresponding scattering. This has already been described using the rational map approach in \cite{Sutcliffe1997}. 
The $D_6$-symmetric monopoles described here corresponds to the
prismatic subgroup $D_{3h}$ of $\operatorname{O}(3)$: this confines the monopoles to lie in a plane, and thus any scattering observed must be planar. 
Note for each value of $\alpha\neq 0$ there are two choices of $\beta$ from the defining equations, and these two branches coalesce where $\beta =0 \Leftrightarrow \alpha = 0$. This gives us a view of scattering from $\alpha=1$ with three initially well-separated 1-monopoles with a choice of sign. They move inwards along the axes of symmetry of the corresponding equilateral triangle through $\alpha = 0$ where the 3-monopoles instantaneously takes the configuration of the axially-symmetric monopole. Here we change branch (i.e. sign of $\beta$), and move back out to $\alpha = 1$ where now because of the change of sign these three well-separated 1-monopoles are deflected by $\pi/3$ radians. Note that as with the planar scattering of 2-monopoles \cite{Atiyah1984}, because of symmetry one cannot associate a given in-going monopole with an out-going one but rather interpret the scattering process as the 3 in-going monopoles splitting into thirds which then recombine to form the out-going monopoles.

\section{\secmath{V_4} Monopoles}\label{V4 Monopoles}

\subsection{Solving for Nahm Data}

For our curve with $V_4$ symmetry the generators of the automorphism group are $(\zeta,\eta) \mapsto (-\zeta, -\eta)$ and $(\zeta, \eta) \mapsto (-1/\zeta, \eta/\zeta^2)$; equivalently these correspond to the rotations $\diag(-1,-1,1)$
and $\diag(-1,1,-1)$ whose product is the earlier $r$. If we impose further the involution $(\zeta, \eta) \mapsto (\zeta, -\eta)$ as a symmetry (the composition of inversion with the anti-holomorphic involution) we restrict to the case of the inversion-symmetric 3-monopoles known in \cite{Houghton1996b}. Here they solve for Nahm matrices given in terms of 3 real-valued functions $f_i$ satisfying $f_1^\prime = f_2 f_3$ (and cyclic), with the corresponding spectral curve being 
\[
\eta^3 + \eta \psquare{\pround{f_1^2 - f_2^2}(\zeta^4 + 1) + (2f_1^2 + 2f_2^2 - 4 f_3^3)\zeta^2} = 0.
\]
We find in Appendix \ref{sec: v4 nahm matrices} that the same procedure, now without imposing the extra symmetry, yields Nahm matrices in terms of 3 complex-valued functions satisfying\footnote{These equations are also found in \cite[(3.57)]{Houghton1997} where they are attributed to \cite{Armstrong1962}; they also appear as the $x$-independent solutions in the description of
$3$-wave scattering \cite[(17) p177]{Novikov1984}. We thank Pol Vanhaecke and Sasha Mikhailov for this latter reference.} 
\begin{equation}\label{complexeuler}
\bar{f_1}^\prime = f_2 f_3 \quad \text{(and cyclic)},
\end{equation} with the corresponding spectral curve being 
\begin{equation}  \label{genv4}
 \eta^3 + \eta \psquare{a(\zeta^4 + 1) + b \zeta^2} + c\zeta (\zeta^4 - 1) = 0,
\end{equation}
where
$$ a=\abs{f_1}^2 - \abs{f_2}^2,\quad
b= 2\abs{f_1}^2 + 2\abs{f_2}^2 - 4 \abs{f_3}^3,\quad
c= 2(f_1 f_2 f_3 - \bar{f_1} \bar{f_2} \bar{f_3}).$$
The Nahm matrices are given by
\begin{align}\label{eqn: V4 Nahm matrices}
    T_1 &= \begin{pmatrix}
        0 & 0 & 0 \\ 0 & 0 & -\bar{f_1} \\ 0 & f_1 & 0 
    \end{pmatrix}, &
    T_2 &= \begin{pmatrix}
        0 & 0 & f_2 \\ 0 & 0 & 0 \\ -\bar{f_2} & 0 & 0 
    \end{pmatrix}, &
    T_3 &= \begin{pmatrix}
        0 & -\bar{f_3} & 0 \\ f_3 & 0 & 0 \\ 0 & 0 & 0
    \end{pmatrix}.
\end{align}

\begin{remark}
    We observe that equations (\ref{complexeuler}) come from the Poisson structure $\pbrace{f_i, \bar{f_j}} = \delta_{ij}$, with Hamiltonian $c/2= f_1 f_2 f_3 - \bar{f_1} \bar{f_2} \bar{f_3}$. This complex extension of the Euler equations is integrable.
\end{remark}

\begin{remark}
    We have not fully used up the gauge symmetry available to us. Namely, if we conjugate the $T_i$ by $U = \operatorname{diag}(u_1, u_2, u_3)$ where $u_j = e^{i \phi_j}$ and $\sum \phi_j = 0$, we get 
    \[
f_1 \mapsto u_3 u_2^{-1} f_1, \quad f_2 \mapsto u_1 u_3^{-1} f_2, \quad f_3 \mapsto u_2 u_1^{-1} f_3,
    \]
    which preserves the form of the equations. 
\end{remark}

A consequence of this remark and the form (\ref{eqn: V4 Nahm matrices}) is that for the $T_i$ to have residues which form an irreducible representation of $\mathfrak{su}(2)$ it is sufficient for the $f_i$ to have simple poles at $s=0,2$.

In order to find a solution we note that 
$a_{ij} = \abs{f_i}^2 - \abs{f_j}^2$ and $c = 2(f_1 f_2 f_3 - \bar{f_1} \bar{f_2} \bar{f_3}) $ are now constants. 
As $c$ is imaginary it will be useful to introduce $\tilde{c}:=-ic$. Setting $F = |f_1|$ we have 
\begin{align*}
    (F^\prime)^2 &= \pbrace{ \psquare{ \pround{f_1 \bar{f_1}}^{1/2}}^\prime}^2 
    = \pbrace{ \frac{1}{2} (f_1 \bar{f_1})^\prime (f_1 \bar{f_1})^{-1/2}}^2 
    = \frac{1}{4} (f_1 f_2 f_3 + \bar{f_1} \bar{f_2} \bar{f_3})^2 F^{-2}, \\
    &= \frac{1}{4} F^{-2} \psquare{(c/2)^2 + 4 |{f_1}|^2 |{f_2}|^2 |{f_3}|^2 } , \\
    &= \frac{1}{4} F^{-2} \psquare{(c/2)^2 + 4 F^2 (F^2 - a_{12}) (F^2 + a_{31})},
\end{align*}
{and so with $G = F^2$ we get} 
\[
(G^\prime)^2 = \frac{1}{4}c^2 + 4G(G - a_{12})(G + a_{31}) ,
\]
which then has solutions in terms of elliptic functions. In terms of the coefficients of the spectral curve we already have $a_{12} = a$, and we can moreover find $a_{31} = \frac{-1}{4}(b  +2a)$, so we can rewrite the equation as 
\begin{align}\label{eqn: ODE for |F|^2}
(\tilde{G}^\prime)^2 = 4 \tilde{G}^3 - g_2 \tilde{G} - g_3 ,
\end{align}
where $\tilde{G} = G - \frac{b + 6a}{12}$, $g_2 = a^2 + \frac{b^2}{12}$, and $g_3 = \frac{b(b^2 - 36a^2)}{216} + \frac{1}{4}\tilde{c}^2$.  Then $\tilde G=\wp$, the Weierstrass $\wp$-function. The $j$-invariant for this elliptic curve is
\begin{align}\label{jinvgenv4}
    j &= 1728 \frac{g_2^3}{g_2^3 - 27 g_3^3} = \frac{(12a^2 + b^2)^3}{\pround{a^6 - \frac{1}{2}a^4 b^2 + \frac{1}{16}a^2 b^4 + \frac{9}{4}a^2 b \tilde{c}^2 - \frac{1}{16}b^3 \tilde{c}^2 - \frac{27}{16}\tilde{c}^4}},
\end{align}
which is precisely that of the quotient of (\ref{genv4}) by the $V_4$ symmetry. We also note that the pull-back of the invariant
differential of this quotient is exactly that needed when discussing the Ercolani-Sinha constraint.

Before going on to solve this completely, let's recall what remains to be shown to get a monopole spectral curve (i.e. to have our Nahm matrices satisfy all the conditions to give Nahm data). We need to have that the $\wp$-function associated to the above elliptic curve has real period 2, but we will be able to impose this by tuning the coefficients. Also, as the right-hand side
of 
\begin{equation}\label{v4wpf1}
   \wp=|f_1|^2-\frac{b+6a}{12} 
\end{equation}
is always real this requires $\wp$ to be real and so to be taken on a rectangular or rhombic lattice. Also for reality we need that 
\[
G(s) = \wp(s) + \frac{b+6a}{12}, \quad G(s)-a_{12} = \wp(s) + \frac{b-6a}{12}, \quad G(s)+a_{31} = \wp(s) - \frac{b}{6}
\]
are always positive. Once we have achieved this we will have regularity in the region $(0,2)$, and so get the right pole structure. The final condition is symmetry about $s=1$, which is enforced on the $\abs{f_j}$ (because $\abs{f_j} \sim \sqrt{\wp}$), and so the remaining Nahm constraint $T_j(s) = T_j(2-s)^T$ becomes simply $f_j(s) = - \bar{f_j}(2-s)$: that is we require $\arg f_j(s) = \pm\pi - \arg f_j(2-s)$.

Indeed writing $f_j = \abs{f_j} e^{i \theta_j}$ we can work out the equations for the angles, using 
\begin{align}\label{inttheta}
    f_j^\prime = \pround{\abs{f_j}^\prime + i \theta_j^\prime \abs{f_j}} e^{i \theta_j} = \pround{\frac{\abs{f_j}^\prime}{\abs{f_j}} + i \theta_j^\prime } f_j \Rightarrow \theta_j^\prime = \frac{1}{i}\psquare{\frac{f_j^\prime}{f_j} - \frac{\abs{f_j}^\prime}{\abs{f_j}}} = \frac{-\tilde{c}}{4 \abs{f_j}^2} .
\end{align}
The $\theta_j$ are thus strictly monotonic (unless $\tilde{c}=0$, in which case they are constant), and symmetry about $s=1$ of $\abs{f_j}$ then necessitates that $\theta_j(s) - \theta_j(1)$ is antisymmetric about $s=1$.

We also have that 
\[
\tilde{c} = 4 \abs{f_1} \abs{f_2} \abs{f_3} \sin(\theta_1 + \theta_2 + \theta_3) = \sqrt{\tilde{c}^2 + 4(G^\prime)^2} \sin(\theta_1 + \theta_2 + \theta_3) . 
\]
At $s=1$ where $G^\prime(s)=0$ we need $\sin(\theta_1 + \theta_2 + \theta_3) = 1$, and by our gauge freedom we can choose $\theta_1(1)=\pi/2=\theta_2(1)$ and so $\theta_3(1) = -\pi/2$, thus enforcing our  condition of symmetry about $s=1$. We then see that the anti-symmetry of $\theta_j(s)-\theta_j(1)$ about $s=1$ enforces the remaining reality condition. We also note that as $|f_j(s)|^2=\wp(s)-c_j:=\wp(u)-\wp(v_j)$ for appropriate $s$
and $v_j = \int_{\infty}^{c_j} \psquare{4 u^3 - g_2 u - g_3}^{-1/2} du$ we have \cite[(6.14.6)]{Lawden1989}
\begin{equation}\label{intwp}
\int \frac{du}{\wp(u)-\wp(v)}=
\frac1{\wp'(v)}\left[ 2 u \zeta(v) +\ln\frac{\sigma(u-v)}{\sigma(u+v)}\right],
\end{equation}
allowing us to find the $\theta_j(s)$ explicitly which is done in (\ref{thetaint}) in Appendix \ref{sec: thetaint}.

It remains to fix the real period of the corresponding elliptic curve. 
We describe two methods. The first makes use of the Jacobi elliptic functions to express the lattice invariants in terms of complete elliptic functions \cite[\S 18.9]{Abramowitz1972}.
We explain this in Appendix \ref{sec: restrictions} in which,
by showing that we may fix the real period, establishes
the following theorem. 

\begin{theorem}\label{thm: V4 monopole spectral curves}
    Given $\alpha \in \mathbb{R}$, $m \in [0,1]$, and $\sgn = \pm1$, define $g_2, g_3$ by $g_2 = 12 \pround{K(m)^2/3}^2 q_1(m)$, $g_3 = 4\pround{{K(m)^2}/{3}}^3(2m-1)q_2(m)$, where 
    \[
    q_1(m) = \left \lbrace \begin{array}{cc} 1 - m + m^2 & \sgn=1, \\ 1 - 16 m + 16 m^2 & \sgn=-1,\end{array} \right. \quad q_2(m)  = \left \lbrace \begin{array}{cc} (m-2)(m+1) & \sgn=1, \\ 2(32m^2 - 32m - 1) & \sgn=-1.\end{array} \right.
    \]
    If $m$ is such that $g_2>0$ and the polynomial $(4-2\alpha)x^3 - g_2 x - g_3$ has a real root $x_\ast$ with $\abs{x_\ast}<\sqrt{g_2/3}$ and  $\operatorname{sgn}(x_\ast) = -\operatorname{sgn}(\alpha)$, then we may solve 
    \[
    a^2 + \frac{b^2}{12} = g_2, \quad \frac{b(b^2 - 36a^2)}{216} + \frac{\tilde{c}^2}{4} = g_3
    \]
    for $a, b, \tilde{c} \in \mathbb{R}$ taking $\alpha = \frac{-27 \tilde{c}^2}{b^3}$. Then 
    \[
    \eta^3 + \eta\psquare{a(\zeta^4 + 1) + b\zeta^2} + i \tilde{c}\zeta (\zeta^4-1) = 0
    \]
    is a monopole spectral curve with $V_4$ symmetry. Moreover the Nahm data is given explicitly in terms of elliptic functions by (\ref{eqn: V4 Nahm matrices}), (\ref{v4wpf1}) and (\ref{thetaint}). 
\end{theorem}

A second approach to fixing the correct real period to give Nahm data is to invert the $j$-invariant (\ref{jinvgenv4})
as done in the earlier $D_6$ case. Though we are unable to invert in terms of a single rational $\alpha$ as with the $D_6$-symmetric monopole, we may use \cite[(4)]{Duke2008} which gives
\[
\tau = i\psquare{\frac{2\sqrt{\pi}}{\Gamma(7/12)\Gamma(11/12)} \frac{{}_2 F_1 (1/12 ,5/12, 1/2 ; x)}{{}_2 F_1 (1/12 ,5/12, 1 ; 1-x)} - 1},
\]
where $x = 1 - \frac{1728}{j} = \frac{(1 - 2\alpha - 3\gamma)^2}{(1+\gamma)^3}$, with $\alpha = - \frac{27 \tilde{c}^2}{b^3}$, $\gamma = \frac{12a^2}{b^2}$. One may then fix the real period of the lattice, which will give solutions consistent with the definition of $x$ for some range of the parameters $\alpha, \gamma$. We investigate one particular restriction of this kind in \S\ref{sec: D4 monopoles}. We remark that \cite{Houghton1997} solved the associated Nahm data only for the (one-parameter) case $\Delta=0$ in which the elliptic curve degenerates and has trigonometric solutions.

\subsection{\secmath{D_4} Monopoles}\label{sec: D4 monopoles}
In \cite{Houghton1996c} a subfamily of (\ref{genv4}) with $D_4$ symmetry was studied. To the existing $V_4$ symmetries is appended the order-$4$ element $(\zeta,\eta)\mapsto(i\zeta, -i\eta)$ (corresponding to the composition of inversion with a rotation of $\pi/2$ in the $xy$-plane). This symmetry then requires $a=0$. 
By a dimension argument we expect the $j$-invariant inversion to yield a geodesic 1-parameter family for the enlarged symmetry group, and this was the case considered in \cite{Houghton1996c}
where the $C_4$ quotient yields an elliptic curve. Placing this curve in our $V_4$ family allows us a different approach to this family of curves. The restriction $a=0$ means that $\frac{1728}{j} = 4\alpha(1-\alpha)$ with $\alpha = - \frac{27 \tilde{c}^2}{b^3}$
and we can then fix the real period via the same approach as for the $D_6$ monopole. The equations we get are 
\begin{align*}
    \frac{b^2}{3} = \frac{1}{4}g_2(1, \tau), \quad 
    \frac{b^3}{27} + 2\tilde{c}^2 = \frac{1}{8} g_3(1, \tau),
\end{align*}
with these being consistent with the definition of $\alpha$ provided $\operatorname{sgn}(g_3(1, \tau)) = \operatorname{sgn}(b)\operatorname{sgn}(1-2\alpha)$. To also have that $\tilde{c}$ is real, we must have $\operatorname{sgn}(b) = - \operatorname{sgn}(\alpha)$ and hence our consistency condition is $\operatorname{sgn}(g_3(1, \tau)) = -\operatorname{sgn}(\alpha)\operatorname{sgn}(1-2\alpha)$. We thus have solutions in the region $\alpha \in (0,1/2)$ if $\operatorname{sgn}(g_3(1, \tau)) <0$, which requires $\tau = -1/\tau(\alpha)$. We can extend this to $\alpha \in (1/2, 1)$ still taking $\tau = -1/\tau(\alpha)$. Moreover, for $\alpha < 0$, we require $g_3(1, \tau) >0$, which can be achieved taking $\tau = \tau(\alpha)$. Finally, for $\alpha>0$, we require $\operatorname{sgn}(g_1(1, \tau)) >0$, achievable with $\tau = -1/\tau(\alpha)$. As such the parameter region in this case is the whole of $\mathbb{R}$. A case-by-case consideration shows that $G$, $G - a_{12}$, $G + a_{31}$ are always positive on the interval $[0,2]$, so we do indeed get Nahm data as desired.

As with the $D_6$-symmetric monopoles we may identify 
 special values of $\alpha$ and the curves they give. A similar analysis gives those found in \cite{Houghton1996c}, namely
\begin{itemize}
    \item $\alpha=\pm\infty$ gives the tetrahedrally-symmetric monopole, 
    \item $\alpha = 0^+, 0^-$ gives three well-separated 1-monopoles and the axially-symmetric monopole respectively,
    \item $\alpha=1/2$ gives the ``twisted figure-of-eight" monopole. Note $\alpha=1/2$ corresponds to the square lattice we saw as distinguished for the $D_6$ monopole.
\end{itemize}
We additionally see the curve with $\alpha=1$ as distinguished in our parametrisation, which gives the curve
\[
\eta^3 - \pi^2 \eta \zeta^2 \pm \frac{i}{\sqrt{27}} \pi^3 \zeta(\zeta^4-1) = 0. 
\]
In terms of the parameters $a, \epsilon$ of \cite{Houghton1996c}, this curve is given by $a = 2\sqrt{2}$, $\epsilon=-1$.

\subsubsection{Scattering}
As such we can now understand our scattering as starting at $\alpha=0^+$ with three well-separated 1-monopoles. As $\alpha$ increases to $\infty$ we have to pick a choice of $\tilde{c}$ continuously (though there is no specific choice at $\alpha = 0^+$ as the map $\zeta \mapsto -\zeta$ which swaps the choice of $\tilde{c}$ is a symmetry of our well separated configuration), and we pass through two distinguished curves, arriving at the tetrahedrally-symmetric monopoles in one orientation. We match that to $\alpha = -\infty$ taking the tetrahedrally-symmetric monopole with the same orientation there, allowing $\alpha$ to then increase up to $0^-$ where it takes the configuration of the axially-symmetric monopole. Here the two branches of $\tilde{c}$ coalesce, we change branch and do the process in reverse.

\section{Conclusion}\label{Conclusion}

In this letter we have begun systematising the classification of charge-3 monopole spectral curves with automorphisms, providing an exhaustive list of candidate curves; we nevertheless expect this list to contain curves that do not correspond to monopole spectral curves. We have also identified how one may use group theory to identify the subset of these candidates that quotient to an elliptic curve. This was done because such curves are amenable to the construction of Nahm matrices in terms of elliptic functions
 using the procedure of \cite{Hitchin1995}. Here 
the imposition of Hitchin's conditions (or equivalently those of Ercolani-Sinha) reduces to questions about the real periods of elliptic functions. Having provided new candidate spectral curves we solved for the Nahm data in two new cases, those of $D_6$ and $V_4$ symmetry.
The latter led us to an integrable system (\ref{complexeuler}) that may be viewed as the complexified Euler equations.
Given Nahm matrices and the corresponding group action what is not yet clear is how to methodically extract from the resulting coupled ODE's the relevant elliptic equations; providing such an understanding would simplify the construction of the solutions to Nahm's equations from the spectral data. This is the reason for our not treating the $C_2$-symmetric monopoles here: in this case we have $13$ coupled ODE's with $7$ conserved quantities.

One can generalise to higher charge several of the viewpoints put forward in this paper. We have seen that compactifying mini-twistor space in $\mathbb{P}\sp{1,1,2}$ and then looking at its image in $\mathbb{P}^3$ a possible charge-$k$ spectral curve is represented by the intersection of the cone and a degree-$k$
hypersurface. There may be value in this viewpoint for providing a candidate list of monopole spectral curves in higher charge. Further, the methods used to calculate the group-signature pairs giving elliptic quotients in genus 4 extends to higher genus, and so may be used to provide candidate spectral curves potentially amenable to solutions in terms of elliptic functions at higher charges. At present, this data has not been computed in the LMFDB, and so a first step would be the tabulation of those results. In the event that such a computation produced too extensive a list we suggest restricting to the case where $\delta (g, G, c)=1, 2$, for which we expect any corresponding monopole spectral curves to be either isolated points in the moduli space or to correspond to geodesic motion respectively, as we conjectured.

Finally, the geometry we introduced here may have applications for the understanding of spectral curves of hyperbolic monopoles. Spectral curves corresponding to hyperbolic monopoles live in the mini-twistor space of hyperbolic space, which is isomorphic to $\mathbb{P}^1 \times \mathbb{P}^1$, and specifically charge-$k$ hyperbolic monopoles are bidegree-$(k,k)$ curves in this surface \cite{Atiyah1984}. As $\mathbb{P}^1 \times \mathbb{P}^1$ is isomorphic to the non-singular quadric in $\mathbb{P}^3$, and bidegree-$(3,3)$ curves in this correspond to the other class of non-hyperelliptic curves classified by Wiman, our work  highlights the potential of classifying certain hyperbolic monopole spectral curves.

\appendix

\section{Hypergeometric Functions and Lattice Invariants}\label{sec: hypergeometric function proofs}

We gather here some of properties of elliptic and related functions used in the text and prove those statements noted in the text. We follow the conventions of \cite[Chapter 23]{dlmf2022}. First we recall the Weierstrass-$\wp$ function is defined by
$$\wp'\sp2=4\wp^3-g_2\wp-g_3=4(\wp-e_1)(\wp-e_2)(\wp-e_3).$$
Here $g_k=g_k(\omega,\omega^\prime)$ are defined by the lattice
$\Lambda:=2\omega\mathbb{Z}+2\omega^\prime\mathbb{Z}$. Let $\tau=\omega^\prime/\omega$.
We have 
 \begin{enumerate}[(i)]
    \item $g_k(\lambda \omega, \lambda \omega^\prime) = \lambda^{-2k} g_k(\omega, \omega^\prime)$, 
     \item given $\begin{psmallmatrix} a & b \\ c & d\end{psmallmatrix} \in \SL(2, \mathbb{Z})$, $g_k(1, (a\tau+b)(c \tau + d)^{-1}) = (c \tau + d)^{2k} g_k(1, \tau)$, 
     \item $\lim_{\operatorname{Im}\tau \to \infty} g_2(1, \tau) = \frac{4 \pi^4}{3}$, $\lim_{\operatorname{Im}\tau \to \infty} g_3(1, \tau) = \frac{8 \pi^6}{27}$. 
     \item When $\tau = i$, $g_2(1, \tau) = \frac{\Gamma(1/4)^8}{256 \pi^2}$, $g_3(1, \tau) = 0$.
    \item When $\tau = e^{2\pi i/3}$, $g_2(1, \tau) = 0$, $g_3(1, \tau) = \frac{\Gamma(1/3)^{18}}{(2\pi)^6}$.
 \end{enumerate}

Our parametrization of the Nahm matrices requires us to know the reality properties of the $\wp$-function. We have
\cite[Theorem 3.16.2]{Jones1987}
\[
g_2, g_3 \in \mathbb{R} \Leftrightarrow \forall z \in \mathbb{C}, \  \wp(\bar{z}; g_2, g_3) = \overline{\wp(z; g_2, g_3)} \Leftrightarrow \Lambda = \overline{\Lambda}.\]
Lattices for which $\Lambda = \overline{\Lambda}$ are called real lattices and they fall into two classes; rectangular lattices ($\omega \in \mathbb{R}$, $\omega^\prime \in i\mathbb{R}$), and rhombic lattices ($\overline{\omega} = \omega^\prime$). The rhombic lattices correspond to $\tau$ being on the boundary of the fundamental domain of the $\SL(2, \mathbb{Z})$ action on the upper half plane while the rectangular
lattices correspond to $\tau$ on the imaginary axis with $\Im(\tau)\ge1$. When restricted to rectangular or rhombic lattices  we can say more about the values of $g_2(1, \tau)$ and $g_3(1, \tau)$. This is done by relating the $g_k$ to the roots $e_i$ of the corresponding cubic equation by 
 \[
 g_2(1, \tau) = 2(e_1^2 + e_2^2 + e_3^2), \quad g_3(1, \tau) = -4 e_1 e_2 e_3. 
 \]
 \begin{enumerate}[(i)]
     \item On a rectangular lattice have $e_i \in \mathbb{R}$ so $g_2 > 0$; further, $g_3 > 0$ if $\abs{\tau} > 1$, $g_3< 0$ if $\abs{\tau}<1$.
     \item On a rhombic lattice, $e_1 \in \mathbb{R}$, $e_2 = \bar{e}_3$, and $\operatorname{sgn}(e_1) = \operatorname{sgn}(g_3)$. 
 \end{enumerate}

\subsection{Properties of \secmath{\tau(\alpha)}}\label{proptau}

In order to use the process of $j$-invariant inversion to impose the correct periodicity constraints and to give the 
limiting behaviours noted in the text  enabling us to find certain distinguished monopoles we require the properties of 
\[
\tau = \tau(\alpha) = i \frac{{}_2 F_1 (1/6 ,5/6, 1 ; 1-\alpha)}{{}_2 F_1 (1/6 ,5/6, 1 ; \alpha)}. 
\]
This is multi-valued when $\alpha < 0$ \cite[15.2.3]{dlmf2022}, with a principal branch $\tau_p$ and second branch $\tau_p + 1$, but for our purposes this difference will not not be important. The specific properties we require are that 
\begin{enumerate}[(i)]
    \item $\forall \alpha \in (0,1),$ $\tau(\alpha) \in i \mathbb{R}_{>0}$, 
    \item $\tau( 0^+) = +i\infty$, $\tau(1/2)=i$, $\tau(1^-) = 0$,
    \item $\forall \alpha < 0$, $\operatorname{Re}(\tau(\alpha)) \equiv 1/2$ mod 1,
    \item $\tau(-\infty) = e^{2\pi i/3}$, $\tau(0^-) = \frac{1}{2} + i \infty$. 
\end{enumerate}
Evaluated at the specific $\tau(\alpha)$ above we find that 
\[
\operatorname{sgn}(g_3(1, \tau(\alpha))) = \left \lbrace \begin{array}{cc}
    1 & \alpha < 1/2,  \\
    -1 & \alpha \in (1/2, 1). 
\end{array}\right . 
\]
Here we provide the necessary definitions and proofs. 

We can understand the behaviour of $\tau$ using known results about hypergeometric functions (see for example \cite[\S 15]{Abramowitz1972}). First, in the region $\alpha \in (0,1)$, we may use the series expression for ${}_2 F_1 (a, b, c; z)$ when $\abs{z} < 1$:
\[
{}_2 F_1 (a, b, c; z) = \sum_{n=0}^\infty \frac{(a)_n (b)_n}{(c)_n} \frac{z^n}{n!},
\]
where $(a)_n$ is the rising Pochhammer symbol 
\[
(a)_n = \left \lbrace \begin{array}{cc}
1     & n=0, \\
a(a+1)\dots(a+n-1) & n \geq 1. 
\end{array}\right . 
\]
This means we have (i) that for all $\alpha \in (0,1)$, $\tau(\alpha) \in i\mathbb{R}_{>0}$. This is important as it makes the lattice rectangular, which forces the Weierstrass $\wp$-function to be real on the real axis \cite[\S 23.5]{dlmf2022}. Moreover, as ${}_2 F_1 (a, b, c; z)$ is increasing in $z\in(0,1)$, $\operatorname{Im}\tau(\alpha)$ is strictly decreasing in $\alpha$. We can calculate the limits to be 
\[
\tau( 0^+) = +i\infty, \quad \tau(1^-) = 0,
\]
so giving (ii).
We may use \cite[15.3.10]{Abramowitz1972} which says that when $\abs{1-z}<1,$ $\abs{\arg(1-z)} < \pi$, 
\[
{}_2 F_1 (a, b, a+b; z) =
\frac{\Gamma(a+b)}{\Gamma(a) \Gamma(b)} \sum_{n=0}^\infty \frac{(a)_n(b)_n}{(n!)^2} \psquare{2 \psi(n+1) - \psi(a+n) - \psi(n+b) - \log(1-z)}(1-z)^n,
\]
where $\psi$ is the digamma function, to understand exactly this limiting behaviour, namely that the divergence is logarithmic. 
We can also highlight a special value in this regions, namely $ \tau(1/2) = i$.  

For $ \alpha \not \in [0,1]$ we no longer have that $\tau$ lies on the imaginary axis, and we would thus need to get a rhombic lattice (that is $\operatorname{Re}\tau = 1/2$) for the reality of $\wp$. Numerical tests suggest that while this happens for $\alpha < 0$ for $\alpha > 1$ we instead get $\operatorname{Re}(-1/\tau) = 1/2$. Indeed we may use \cite[15.10.29]{dlmf2022} to say\footnote{We are very grateful to Adri Olde Daalhuis for this argument.} 
\begin{align*}
{}_2 F_1 (1/6 ,5/6, 1 ; 1-\alpha) &= e^{5 \pi i/6} \frac{\Gamma(1) \Gamma(1/6)}{\Gamma(1) \Gamma(1/6)}{}_2 F_1 (1/6 ,5/6, 1 ; \alpha) \\
&\phantom{=} + e^{-\pi i/6} \frac{\Gamma(1) \Gamma(1/6)}{\Gamma(5/6)\Gamma(1/3)} 
\alpha^{-1/6}  {}_2 F_1 (1/6 ,1/6, 1/3 ; 1/\alpha), \\
&= e^{5 \pi i/6}{}_2 F_1 (1/6 ,5/6, 1 ; \alpha) \\&\phantom{=} + (-\alpha)^{-1/6} \frac{ \Gamma(1/6)}{\Gamma(5/6)\Gamma(1/3)} {}_2 F_1 (1/6 ,1/6, 1/3 ; 1/\alpha),
\end{align*}
and hence when $\alpha < 0$ (and taking the principal branch of the hypergeometric function) we get 
\[
\tau(\alpha) = i \psquare{e^{5\pi i/6} + T(\alpha)}
\]
with 
\[
T(\alpha) = (-\alpha)^{-1/6} \frac{ \Gamma(1/6)}{\Gamma(5/6)\Gamma(1/3)} \frac{{}_2 F_1 (1/6 ,1/6, 1/3 ; 1/\alpha)}{{}_2 F_1 (1/6 ,5/6, 1 ; \alpha)} \in \mathbb{R}. 
\]
This means $\operatorname{Re}(\tau(\alpha)) \equiv 1/2 \mod 1$,
which yields (iii).

To get the asymptotics as $\alpha \to -\infty$, we use \cite[15.3.7]{Abramowitz1972}
\begin{align*}
{}_2 F_1 (a, b, c; z) &= \frac{\Gamma(c) \Gamma(b-a)}{\Gamma(b)\Gamma(c-a)} (-z)^{-a} {}_2 F_1 (a, a+1-c, a+b-1; z^{-1}) \\
&\phantom{=} + \frac{\Gamma(c) \Gamma(a-b)}{\Gamma(a)\Gamma(c-b)} (-z)^{-b} {}_2 F_1 (b, b+1-c, b+a-1; z^{-1}) . 
\end{align*}
Taking $\alpha = - \epsilon^{-1}$, this gives that as $\epsilon \to 0^+$, 
\[
{}_2 F_1 (a, b, c; -\epsilon^{-1}) \sim \frac{\Gamma(2/3)}{\Gamma(5/6)^2} \epsilon^{1/6}, \quad {}_2 F_1 (a, b, c; 1+\epsilon^{-1}) \sim \frac{\Gamma(2/3)}{\Gamma(5/6)^2} (-\epsilon)^{1/6},
\]
and so $\tau(-\infty) = e^{2 \pi i /3} = \frac{-1}{2} + \frac{i\sqrt{3}}{2}$. 
To get the remaining asymptotics of (iv), as $\alpha \to 0^-$ we write $\alpha = -\epsilon$. Then 
\[
{}_2 F_1 (a, b, c; -\epsilon) \sim 1, \quad {}_2 F_1 (a, b, c; 1+\epsilon) \sim \frac{-\Gamma(1)}{\Gamma(1/6)\Gamma(5/6)}\log(-\epsilon) = \frac{-1}{2\pi}(i \pi  + \log \epsilon),
\]
and so $\tau(0^-) = \frac{1}{2} + i \infty$.

To get the asymptotics as $\alpha \to 1^+$ we recognise that $\tau(1-\alpha) = -1/\tau(\alpha)$ and so $-1/\tau(1^+) = \frac{1}{2} + i \infty$. 
Finally to get the asymptotics as $\alpha \to \infty$ we do the same, so 
$-1/\tau(\infty) = \frac{1}{2} + \frac{i\sqrt{3}}{2}$.

\subsection{Check of Consistency}\label{consistency}
 We see that $\alpha_2$ must be the same sign as $\alpha$ to get $\beta \in \mathbb{R}$. Moreover, as $g_2(1,\tau)>0$ because $\alpha_2$ is real, we can check that
 \begin{align*}
 \frac{1}{2} \psquare{\frac{1}{27} \alpha_2^3 - \frac{1}{8} g_3(1, \tau)} &= \frac{1}{2} \psquare{ \frac{\operatorname{sgn}(\alpha_2)}{27} \pround{3 g_2(1,\tau)/4}^{3/2} - \frac{1}{8} g_3(1, \tau)}, \\
 &= \frac{\operatorname{sgn}(\alpha_2) g_2(1,\tau)^{3/2}}{16 \sqrt{27}} \psquare{ 1 - \frac{\operatorname{sgn}(g_3(1, \tau))}{\operatorname{sgn}(\alpha_2)}\sqrt{\frac{27 g_3^2}{g_2^3}}}, \\
 &= \frac{\operatorname{sgn}(\alpha_2)\pround{4 \alpha_2^2/3}^{3/2}}{16 \sqrt{27}} \psquare{ 1 - \frac{\operatorname{sgn}(g_3(1, \tau))}{\operatorname{sgn}(\alpha_2)}\pround{1 - \frac{1728}{j}}^{1/2}} , \\
 &= \frac{\alpha_2^3}{2 \times 27} \psquare{ 1 - \frac{\operatorname{sgn}(g_3(1, \tau))}{\operatorname{sgn}(\alpha_2)}\pround{1 - 4 \alpha(1-\alpha)}^{1/2}}, \\
 &= \frac{\alpha_2^3}{2 \times 27} \psquare{ 1 - \frac{\operatorname{sgn}(g_3(1, \tau))}{\operatorname{sgn}(\alpha_2)\operatorname{sgn}(1-2\alpha)}(1-2\alpha)}, \\
 &= \frac{\alpha \alpha_2^3}{27} \quad \text{ if } \quad \operatorname{sgn}(g_3(1, \tau)) = \operatorname{sgn}(\alpha_2)\operatorname{sgn}(1-2\alpha).
 \end{align*}
 Hence the two equations are consistent, provided the stated sign condition holds, or if $\alpha = 0$. 
 
\subsection{The Theta Integration} \label{sec: thetaint} We note that $|f_j(s)|^2=\wp(s)-c_j:=\wp(s)-\wp(v_j)$ doesn't fix the sign of $v_j$ for $\wp(\pm v_j)=c_j$. We fix the sign as follows. First observe that 
$$(\tilde{G}^\prime(s))^2=({G}^\prime(s))^2=
\frac{1}{4}c^2 + 4[\wp(s)-\wp(v_1)][\wp(s)-\wp(v_2)][\wp(s)-\wp(v_3)],$$
and so $\wp^{\prime\,2}(v_i)=c^2/4$; we fix the sign so that
$\wp^{\prime}(v_i)=c/2=i\tilde{c}/2$. Further consider the elliptic function
$\wp^{\prime}(s)-c/2$ with three zeros (at $s\in\{v_1,v_2,v_3\}$) and three poles (at $s=0$). Then with the base of the Abel-Jacobi map at $s=0$ (as is standard) we have that
$\sum_i v_i$ is a lattice point. Also observe that
$$\zeta(v_i)+\zeta(v_j)=\zeta(v_i+v_j).$$

We find from
(\ref{inttheta},\ref{intwp}) that
\begin{equation}
\label{thetaint}   
\theta_i(s)
:=
\theta_i(1)+
i \left[ s \zeta(v_j) +\frac12\ln\frac{\sigma(s-v_j)\sigma(1+v_j)}{\sigma(s+v_j)\sigma(1-v_j)} \right],
\end{equation}
where $\theta_i(1)$ is a constant of integration and chosen as described in the text.
Then $\theta_i(-s)-\theta_i(1)=-\left[\theta_i(s)-\theta_i(1)\right]$ is anti-symmetric as required. 
Using the Legendre relation we find that $\sin(\theta_1+\theta_2+\theta_3)$ is periodic in $s$ as
required for consistency.

\subsection{Restrictions on Elliptic Function Parameters}
\label{sec: restrictions}

Here we prove Theorem \ref{thm: V4 monopole spectral curves}. 
Given the discriminant $\Delta$ for the cubic defining $\wp$ \cite[\S 18.9]{Abramowitz1972} gives equations for the lattice invariants in terms of complete elliptic functions.
With our earlier definitions, $g_2 = 12 \pround{K(m)^2/3}^2 q_1(m)$ and $g_3 = 4\pround{{K(m)^2}/{3}}^3(2m-1)q_2(m)$ then with $\sgn=\mathop{\mathrm{sgn}}(\Delta)$ we have for $\Delta > 0$
\[
g_2 = 12\pround{\frac{K^2}{3 \omega_1}}^2\pround{1 - m + m^2}, \quad g_3 = 4\pround{\frac{K^2}{3 \omega_1}}^3 (m-2)(2m-1)(m+1);
\]
whereas for $\Delta < 0$
\[
g_2 = 12\pround{\frac{K^2}{3 \omega_2}}^2\pround{1 - 16m + 16m^2}, \quad g_3 = 8\pround{\frac{K^2}{3 \omega_2}}^3 (2m-1) (32m^2 - 32m - 1).
\]
Here $m=k^2 \in (0,1)$ is the argument of $K$, the underlying lattice has periods $2 \omega$, $2\omega^\prime$, $\omega_1 = \omega$ and $\omega_2 = \omega + \omega^\prime$. Fixing 2 as a period of the lattice, and that the lattice is real, sets $\omega_1 = 1$ for  $\Delta >0$ and $\omega_2=1$ for $\Delta <0$. Observe that for $\sgn(\Delta)=\pm1$ that $g_2(m)$ takes its minimum value at $m=1/2$ while for $m\in(0,1/2)$ we have $g_3(m)>0$.

Our elliptic curve gave the equations $a^2 + \frac{b^2}{12} = g_2$, $\frac{b(b^2 - 36a^2)}{216} + \frac{1}{4}\tilde{c}^2 = g_3$. These equations are underdetermined, but we may substitute for $a^2$ and take $\alpha = {-27 \tilde{c}^2}/{b^3}$ to find 
\begin{equation}\label{cubicalpha}
(4-2\alpha) \tilde{b}^3 - g_2 \tilde{b} - g_3 = 0,
\end{equation}
where $\tilde{b} = \frac{b}{6}$. The discriminant of this cubic is 
$$
\Delta_\alpha(m) = 4(4-2\alpha) g_2^3 - 27(4-2\alpha)^2 g_3^2 = 4(4-2\alpha)\psquare{g_2^3 -27(1-\alpha/2)g_3^2}.
$$
Note $\Delta_0 = \Delta$. 
For a given generic value of $\alpha$ in some region we may solve (\ref{cubicalpha}) , determining $b$, $\tilde{c}$ and $a$ in turn. 

In order to get Nahm data, we require that $b$, $\tilde{c}$, and $a$ are real. We know that this cubic has real coefficients, and so there will always be a real root of the cubic. To get reality of $\tilde{c}$, we need that this real root $\tilde{b}$ satisfies $\operatorname{sgn}(\tilde{b}) = -\operatorname{sgn}(\alpha)$, and for reality of $a$ we need $\abs{\tilde{b}} \leq g_2/3$. Necessary conditions to find such solutions are as follows. 

First consider $\Delta > 0$. Then
$g_2 > 0$ and $g_3$ is monotonically decreasing for $m\in(0,1)$ with $\operatorname{sgn}(g_3) = -\operatorname{sgn}(m - 1/2)$. We have the following properties:
\begin{enumerate}[(i)]
    \item If $\alpha > 2$ the discriminant $\Delta_\alpha(m)<0$. Then (\ref{cubicalpha}) has one real root whose sign is opposite that of  $g_3$. 
    Now $\operatorname{sgn}(\tilde{b}) = -\operatorname{sgn}(\alpha)<0$ is opposite that of $g_3$; hence we require
    $g_3>0$ and so $m\in(0,1/2)$.
    \item For $\alpha \in(0,2)$ the discriminant $\Delta_\alpha(m)>0$ upon comparison with
    $\Delta=g_2^3 - 27 g_3^2 > 0$.
    Then, because the sum of the roots is zero, they cannot all be the same sign.
    \item When $\alpha < 0$, 
    from the derivative of the cubic we know it will have a local maxima and minima at $\tilde{b} = \pm\sqrt{\frac{g_2}{3(4-2\alpha)}}$; it is the minima when the sign is positive. Recalling that we require a root with sign $\sgn \tilde{b} = -\sgn\alpha = 1$, the local minima must be non-positive, and the value at this $\tilde{b}$ is $-\frac{2}{3}\tilde{b}g_2 - g_3$. As the value at this minima is monotonically increasing for $m>1/2$, and negative at $m=1/2$, then the value at the minima is negative for all $m < m_\ast$, the value for which is it zero. Solving, one gets the condition $\Delta_\alpha(m_\ast)=0$, taking the root greater than $1/2$.   
    
\end{enumerate}
Therefore necessary conditions for a real root of the right sign to exist for
$\Delta>0$ are that
\begin{itemize}
    \item if $\alpha > 2$, $m  < 1/2$,
    \item if $\alpha \in (0, 2)$, any $m$ is valid
    \item if $\alpha < 0$, $m <m_2(\alpha)$, where $m_2$ is the root of $\Delta_\alpha(m)=0$ in $(1/2,1)$.
\end{itemize}
To get Nahm data we require that this real root is bounded in magnitude by $\sqrt{g_2/3}$, with the case that it is equal corresponding to $a=0$, i.e. to the $D_4$ monopoles of \cite{Houghton1996c}. Figures showing these parameter regions are given in Figure \ref{fig: valid parameter regions}. 


In the case $\Delta < 0$, in order to get real roots of the right sign one analogously gets restrictions on $m$ relative to $\alpha$ such that
\begin{itemize}
    \item if $\alpha < 0$, $m < 1/2$, 
    \item if $\alpha \in (0,2)$, $m > 1/2$ or $m < m_1(\alpha)$, defined to be the root $<1/2$ of the polynomial $\Delta_\alpha(m)=0$. 
    \item if $\alpha > 2$, $m < m_2(\alpha)$, now defined to be the root $>1/2$ of the polynomial $\Delta_\alpha(m)=0$.
\end{itemize}
Fixing the size of the root in this case requires more work, complicated by the fact that $g_2$ is real only if $\abs{m-1/2} > \sqrt{3}/4$. Using explicitly formulas for the roots $\tilde{b}$ from Cardano's formula one can achieve explicit bounds, but here we omit these. In practice, when using this approach to plot monopoles, numerical methods can be used to find the appropriate $m$ region for a given $\alpha$, as done to generate Figure \ref{fig: valid parameter regions}. 

\begin{figure}
    \centering
     \begin{subfigure}[c]{0.49\textwidth}
         \centering
         \includegraphics[width=\textwidth]{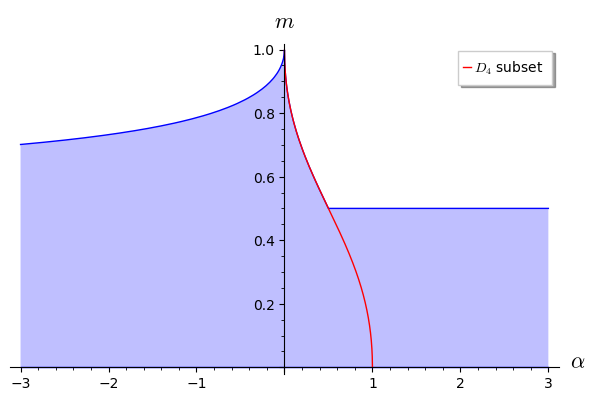}
         \caption{$\Delta > 0$}
         \label{fig: vpr+1}
     \end{subfigure}
     \begin{subfigure}[c]{0.49\textwidth}
         \centering
         \includegraphics[width=\textwidth]{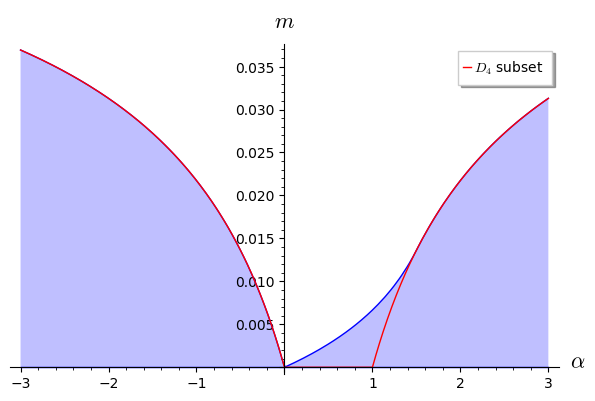}
         \caption{$\Delta < 0$}
         \label{fig: vpr-1}
     \end{subfigure}
     \caption{Valid parameter regions for $V_4$ monopoles, with the subset corresponding to $D_4$ monopoles highlighted}
     \label{fig: valid parameter regions}
\end{figure}

Note that, for certain admissible $\alpha, m$ there may be two possible monopoles because two roots of the cubic defining $b$ satisfy the required conditions. Numerical investigations indicates that this phenomenon only occurs for $\Delta > 0 $. We plotted two examples of this, seen in Figure \ref{fig: plots of monopoles for fixed params} to investigate the difference in the associated monopoles. These regions stitch together so as to make the moduli space connected. 

\begin{figure}
    \centering
     \begin{subfigure}[c]{0.49\textwidth}
         \centering
         \includegraphics[width=\textwidth]{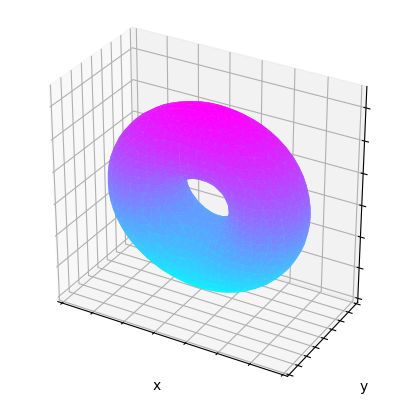}
         \caption{$k=0.45$, $\alpha = 0.2$, $\Delta > 0$, $b = -3.21$}
     \end{subfigure}
     \begin{subfigure}[c]{0.49\textwidth}
         \centering
         \includegraphics[width=\textwidth]{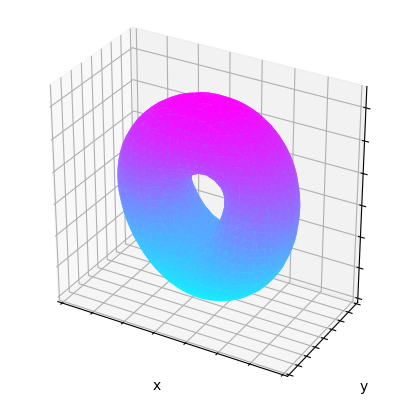}
         \caption{$k=0.45$, $\alpha = 0.2$, $\Delta > 0$, $b = -7.19$}
     \end{subfigure}
     \begin{subfigure}[c]{0.49\textwidth}
         \centering
         \includegraphics[width=\textwidth]{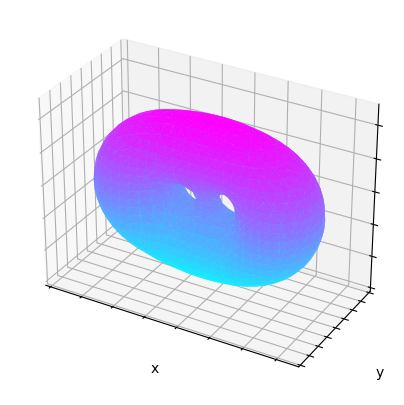}
         \caption{$k=0.77$, $\alpha = -2.0$, $\Delta > 0$, $b=1.42$}
     \end{subfigure}
     \begin{subfigure}[c]{0.49\textwidth}
         \centering
         \includegraphics[width=\textwidth]{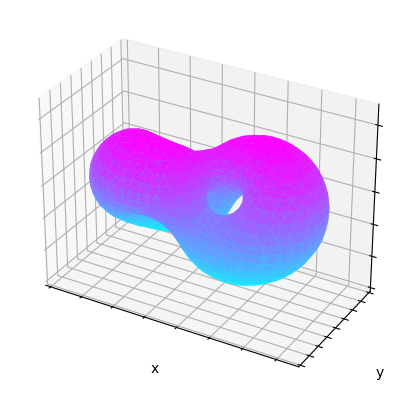}
         \caption{$k=0.77$, $\alpha = -2.0$, $\Delta > 0$, $b=7.24$}
     \end{subfigure}
     \caption{Comparison of two pairs of monopoles with equal values of $\alpha, m$. taking $\mathcal{E}=0.17$}
     \label{fig: plots of monopoles for fixed params}
\end{figure}

\section{Initial Computation of \secmath{D_6} Nahm Matrices}\label{sec: initial computation of nahm matrices}

Here we will take the procedure introduced in \cite{Hitchin1995}
and developed in \cite{Houghton1996,Houghton1996b,Houghton1996c}
and apply this to the $C_k$ symmetric monopole. In brief, this procedure uses representation theory to construct Nahm matrices with a given symmetry group $G \leq \operatorname{SO}(3)$. This is done by establishing an isomorphism between the $\operatorname{SO}(3)$-representation space $\mathbb{R}^3 \otimes \mathfrak{su}(k)$ containing the matrix triples and a $\operatorname{SO}(3)$-representation space described by homogeneous bivariate polynomials acted on naturally via the isomorphism $\operatorname{SO}(3) \cong \operatorname{PSU}(2)$. On the latter the action of the symmetry group is well understood, and so one can take a polynomial $Q$ invariant under $G$ and return via the isomorphism a triple of matrices $(S_i)$ also invariant under the symmetry. The number of possible matrix triples one can get from the input polynomial is determined by the representation theory, and there is always a triple $(\rho_i)$ invariant under the action of $\operatorname{SO}(3)$ corresponding to the polynomial 1. The matrices returned by this isomorphism need not be anti-Hermitian, but they can be fixed to be so. 

We describe the procedure for general $k$ before applying this in the $k=3$ context. The steps of the procedure are:
\begin{enumerate}
    \item Take the input polynomials to be $Q_l = \zeta_0^l \zeta_1^l, i=l, \dots, k$ and $Q_{k+1} = \zeta_0^{2k} - \zeta_1^{2k}$. 
    \item Construct the invariant matrix triples $(\rho_i), (S_i^{(j)})$ corresponding to the polynomials 1, $Q_l$ respectively, and scale them so they are all anti-Hermitian. The degree of the input polynomial $d$ determines how many direct summands of the $(d+1)$-dimensional irreducible representation space of $\operatorname{SO}(3)$ there are in $\mathbb{R}^3 \otimes \mathfrak{gl}(k)$ when decomposed into irreducibles. There are $3k-2$ invariant $S$-vectors and so variables $y_j$ associated with them; together with $(\rho_i)$ and the associated variable $x$ we have $3k-1$ variables.

    \item Set now $T_i = x \rho_i + \sum_j S_i^{(j)}$, where $x, y_j$ are real functions and so $T_i$ is anti-Hermitian.
    \item Diagonalise the matrix $T_3$ with a unitary matrix $U$ whose columns are the normalised eigenvectors of $T_3$, that is construct $U^{-1} T_3 U$. As $T_3$ is anti-hermitian and linear in the invariant vectors, the diagonal entries which are the eigenvalues will be pure-imaginary, linear in $\pbrace{x, y_j}$. 
    \item Now from \cite{Braden2011}, conjugating by the same unitary matrix will give 
    \[
    U^{-1} (T_1 + i T_2)U = \sum_{j=1}^k \alpha_j E_{j, j+1}, \quad  U^{-1} (T_1 - i T_2)U = \sum_{j=1}^k -\bar{\alpha}_j E_{j+1, j},
    \]
    for some $\alpha_1, \dots \alpha_n \in \mathbb{C}$.
    \item Writing $\alpha_j = r_j e^{i\phi_j}$ (as generically $\alpha_j \neq 0$) and solving 
    \begin{align*}
    \phi_j + \theta_{j+1} - \theta_j &= 0 , j=1, \dots k-1, \quad
    \sum_j \theta_j = 0,
    \end{align*}
    for $\theta_1, \dots, \theta_k \in \mathbb{R}$, then conjugating by the unitary matrix $D = \operatorname{diag}(e^{i\theta_1}, \dots, e^{i \theta_k})$ preserves the form of $T_3$, but acts to make each $\alpha_j$ real for $j=1, \dots, k-1$ (as it multiplies $\alpha_j$ by $e^{i(\theta_{j+1}-\theta_j)}$). 
    The effect on $\alpha_k$ is to multiply this by $e^{i(\theta_1 - \theta_k)} = e^{i(\phi_1 + \dots + \phi_{k-1})} = \prod_{j=1}^{k-1} (\alpha_j/r_j)$. After quotienting by this action the number of independent variables we have is $3k-2 + 1 - (k-1) = 2k$.

\end{enumerate}

We now apply the algorithm in the case $k=3$. This yields Nahm matrices 
\begin{align*}
    T_1 &= \left(\begin{array}{rrr}
0 & i y_{0} & i y_{1} + y_{5} + i y_{6} \\
i y_{0} & 0 & -2 x - y_{2} + i y_{3} \\
i y_{1} - y_{5} + i y_{6} & 2 x + y_{2} + i y_{3} & 0
\end{array}\right), \\
T_2 &= \left(\begin{array}{rrr}
i y_{0} & 0 & 2 x + y_{2} - i y_{3} \\
0 & -i y_{0} & i y_{1} + y_{5} + i y_{6} \\
-2 x - y_{2} - i y_{3} & i y_{1} - y_{5} + i y_{6} & 0
\end{array}\right), \\
T_3 &= \left(\begin{array}{rrr}
i y_{1} + i y_{4} - \frac{2}{3} i y_{6} & -2 x + 2 y_{2} & 0 \\
2 x - 2 y_{2} & i y_{1} + i y_{4} - \frac{2}{3} i y_{6} & 0 \\
0 & 0 & -2 i y_{1} + i y_{4} + \frac{4}{3} i y_{6}
\end{array}\right) . 
\end{align*}
with accompanying ODEs in $8$ real variables 
\begin{align*}
x^\prime &= 2 x^{2} - \frac{1}{3} y_{0}^{2} - \frac{5}{6} y_{1}^{2} - \frac{1}{2} y_{2}^{2} + \frac{1}{6} y_{3}^{2} + \frac{1}{6} y_{5}^{2} + \frac{5}{6} y_{6}^{2}, \\
y_0^\prime &= -4 x y_{0} + 4 y_{0} y_{2}, \\
y_1^\prime &= -4 x y_{1} - \frac{16}{5} y_{1} y_{2} - \frac{6}{5} y_{3} y_{5} - \frac{6}{5} y_{2} y_{6}, \\
y_2^\prime &=  \frac{2}{3} y_{0}^{2} - \frac{4}{3} y_{1}^{2} - 2 x y_{2} - y_{2}^{2} - \frac{1}{3} y_{3}^{2} - \frac{1}{3} y_{5}^{2} - y_{1} y_{6} + \frac{1}{3} y_{6}^{2}, \\
y_3^\prime &= 2 x y_{3} - 2 y_{2} y_{3} - 3 y_{1} y_{5} + 2 y_{5} y_{6}, \\
y_4^\prime &= 0,\\
y_5^\prime &= -3 y_{1} y_{3} + 2 x y_{5} - 2 y_{2} y_{5} + 2 y_{3} y_{6}, \\
y_6^\prime &= -\frac{9}{5} y_{1} y_{2} + \frac{6}{5} y_{3} y_{5} + 6 x y_{6} + \frac{6}{5} y_{2} y_{6}.
\end{align*}
The associated spectral curve is 
\[
\eta^3 + \alpha_1 \eta^2 \zeta + \alpha_2 \eta \zeta^2 + \alpha_3 \zeta^3 + \beta \zeta^6 - \bar{\beta}=0
\]
where 
\begin{align*}
\alpha_1 &= -6 y_{4}, \\
\alpha_2 &= 4 y_{0}^{2} - 8 y_{1}^{2} + 48 x y_{2} - 12 y_{2}^{2} + 4 y_{3}^{2} + 12 y_{4}^{2} + 4 y_{5}^{2} + 24 y_{1} y_{6}  - \frac{4}{3} y_{6}^{2}, \\
\alpha_3 &= -160 x^{2} y_{1} + 16 y_{0}^{2} y_{1} + 8 y_{1}^{3} + 128 x y_{1} y_{2} - 40 y_{1} y_{2}^{2} - 8 y_{1} y_{3}^{2} - 8 y_{0}^{2} y_{4} \\
&\phantom{=}  + 16 y_{1}^{2} y_{4} - 96 x y_{2} y_{4} + 24 y_{2}^{2} y_{4} - 8 y_{3}^{2} y_{4} - 8 y_{4}^{3} - 32 x y_{3} y_{5} + 32 y_{2} y_{3} y_{5}  \\
&\phantom{=} - 8 y_{1} y_{5}^{2} - 8 y_{4} y_{5}^{2} - \frac{32}{3} y_{0}^{2} y_{6} - \frac{128}{3} y_{1}^{2} y_{6} - 32 x y_{2} y_{6} + 80 y_{2}^{2} y_{6} + \frac{16}{3} y_{3}^{2} y_{6} \\
&\phantom{=} - 48 y_{1} y_{4} y_{6} + \frac{16}{3} y_{5}^{2} y_{6} + 24 y_{1} y_{6}^{2} + \frac{8}{3} y_{4} y_{6}^{2} + \frac{16}{27} y_{6}^{3}, \\
\beta &= -16 x^{2} y_{0} + 4 y_{0} y_{1}^{2} - 16 x y_{0} y_{2} - 4 y_{0} y_{2}^{2} + 8 j y_{0} y_{1} y_{3} - 4 y_{0} y_{3}^{2} - 16 i x y_{0} y_{5}  \\
&\phantom{=} - 8 i y_{0} y_{2} y_{5} + 4 y_{0} y_{5}^{2} + 8 y_{0} y_{1} y_{6} + 8 i y_{0} y_{3} y_{6} + 4 y_{0} y_{6}^{2}.
\end{align*}
In order to make the variables real we have imposed the anti-Hermiticity condition required of the Nahm matrices at the beginning, by making the invariant vectors corresponding to each variable anti-Hermitian.

We may consistently set $y_3 = 0 = y_5 $, which we may view as using the conjugation action of diagonal matrices $\operatorname{diag}(e^{i\theta_1}, e^{i\theta_2}, e^{i\theta_3})$, $\theta_1 + \theta_2 + \theta_3 = 0$. This leaves us with the $2 \times 3 = 6$ real variables we would expect to have from the corresponding Toda. Note that because $\alpha_1^\prime=0$, the centre of mass of the Toda system is already fixed. Moreover, we may centre to consistently set $y_4=0$, and so we now have the equations in the remaining 5 variables as 
\begin{align*}
x^\prime &= 2 x^{2} - \frac{1}{3} y_{0}^{2} - \frac{5}{6} y_{1}^{2} - \frac{1}{2} y_{2}^{2} + \frac{5}{6} y_{6}^{2}, \\
y_0^\prime &= -4 x y_{0} + 4 y_{0} y_{2}, \\
y_1^\prime &= -4 x y_{1} - \frac{16}{5} y_{1} y_{2} - \frac{6}{5} y_{2} y_{6}, \\
y_2^\prime &=  \frac{2}{3} y_{0}^{2} - \frac{4}{3} y_{1}^{2} - 2 x y_{2} - y_{2}^{2} - y_{1} y_{6} + \frac{1}{3} y_{6}^{2}, \\
y_6^\prime &= -\frac{9}{5} y_{1} y_{2} + 6 x y_{6} + \frac{6}{5} y_{2} y_{6},
\end{align*}
with conserved quantities
\begin{align*}
\alpha_2 &= 4 y_{0}^{2} - 8 y_{1}^{2} + 48 x y_{2} - 12 y_{2}^{2} + 24 y_{1} y_{6} - \frac{4}{3} y_{6}^{2}, \\
\alpha_3 &= -160 x^{2} y_{1} + 16 y_{0}^{2} y_{1} + 8 y_{1}^{3} + 128 x y_{1} y_{2} - 40 y_{1} y_{2}^{2}   \\
&\phantom{=}  - \frac{32}{3} y_{0}^{2} y_{6} - \frac{128}{3} y_{1}^{2} y_{6} - 32 x y_{2} y_{6} + 80 y_{2}^{2} y_{6}  \\
&\phantom{=} + 24 y_{1} y_{6}^{2}  + \frac{16}{27} y_{6}^{3}, \\
\beta &= -16 x^{2} y_{0} + 4 y_{0} y_{1}^{2} - 16 x y_{0} y_{2} - 4 y_{0} y_{2}^{2}  \\
&\phantom{=}  + 8 y_{0} y_{1} y_{6}  + 4 y_{0} y_{6}^{2}.
\end{align*}

At this stage the resulting ODEs are somewhat opaque and we may  use the connection to Toda to clarify. Following the steps of the procedure from \cite{Braden2011} outlined earlier we may put the Nahm Lax pair in Toda form, namely with 
\begin{align*}
    T_1 + i T_2 &= \left(\begin{array}{rrr}
0 & -2 \sqrt{2} x - \sqrt{2} y_{1} - \sqrt{2} y_{2} - \sqrt{2} y_{6} & 0 \\
0 & 0 & 2 \sqrt{2} x - \sqrt{2} y_{1} + \sqrt{2} y_{2} - \sqrt{2} y_{6} \\
2 y_{0} & 0 & 0
\end{array}\right), \\
T_1 - i T_2 &= \left(\begin{array}{rrr}
0 & 0 & -2 y_{0} \\
2 \sqrt{2} x + \sqrt{2} y_{1} + \sqrt{2} y_{2} + \sqrt{2} y_{6} & 0 & 0 \\
0 & -2 \sqrt{2} x + \sqrt{2} y_{1} -\sqrt{2} y_{2} + \sqrt{2} y_{6} & 0
\end{array}\right), \\
-2iT_3 &= \left(\begin{array}{rrr}
-4 x + 2 y_{1} + 4 y_{2} -\frac{4}{3} y_{6} & 0 & 0 \\
0 & -4 y_{1} + \frac{8}{3} y_{6} & 0 \\
0 & 0 & 4 x + 2 y_{1} -4 y_{2} -\frac{4}{3} y_{6}
\end{array}\right).
\end{align*}
This gives us variables 
\[
\begin{split}
    a_0 &= 2 y_0 , \quad
    a_1 = -2 \sqrt{2} x - \sqrt{2} y_{1} - \sqrt{2} y_{2} - \sqrt{2} y_{6}, \quad
    a_2 = 2 \sqrt{2} x - \sqrt{2} y_{1} + \sqrt{2} y_{2} - \sqrt{2} y_{6}, \\
    b_1 &= 4 x - 2 y_{1} - 4 y_{2} + \frac{4}{3} y_{6}, \quad
    b_2 = 4 y_{1} - \frac{8}{3} y_{6}, \quad
    b_3 = -4 x - 2 y_{1} + 4 y_{2} + \frac{4}{3} y_{6}.
\end{split}
\]
These variables are the Flaschka coordinates for the periodic Toda system. (Any 6-tuple satisfying $\sum_i b_i = 0$ gives valid $x, y_j$.) In these new variables we have 
\begin{equation}
\begin{split}
    a_0^\prime =& \frac{1}{2}a_0 ( b_3 - b_1),\quad
    a_1^\prime = \frac{1}{2} a_1  (b_1 - b_2), \quad
    a_2^\prime = \frac{1}{2} a_2 ( b_2 - b_3),\\
    b_1^\prime =& a_1^2 - a_0^2,\quad
    b_2^\prime = a_2^2 - a_1^2 ,\quad
    b_3^\prime = a_0^2 - a_2^2 ,\quad
\end{split}
\end{equation}
together with the constants
\begin{align*}
    \alpha_2 &= b_{1} b_{2} + b_{1} b_{3} + b_{2} b_{3} + a_{0}^{2} + a_{1}^{2} + a_{2}^{2}, \quad
    \alpha_3 = b_{1} b_{2} b_{3} + b_{1} a_{2}^{2} + b_{2} a_{0}^{2} + b_{3} a_{1}^{2}, \quad
    \beta = a_0 a_1 a_2.
\end{align*}

At this stage we have $6$ variables and 3 constraints.
One could in principle solve these explicitly using the fact that the flow linearises on the Jacobian of the associated hyperelliptic curve as in \cite[Theorem 5.1]{Moerbeke1976}.
Are simplifications possible?  
We may use Gr\"obner bases in Sage to utilise the constants $\alpha_2,\alpha_3, 0 = \sum b_i$ to eliminate the $b_i$, and we get the equations described in the text,
\begin{align*}
0 &= \sum_{i=0}^2 a_i^2 - \alpha_2 - \frac{1}{3}(d_1^2 + d_1 d_2 + d_2^2), \\
0 &= a_1^2 d_2 - a_2^2 d_1 + \alpha_3 + \frac{1}{3}\alpha_2(d_1 - d_2) + \frac{1}{27}(d_1 - d_2)^3 ,
\end{align*}
where we have introduced $d_i = \frac{2a_i^\prime}{a_i}$. This in principle is the maximal reduction one can achieve with the variables provided when the $\alpha_i$ and $\beta$ are generic.

One simplification which can be achieved is by attempting to make the second equation a polynomial in $d_1 - d_2$. To do this we would need $a_1^2 = a_2^2$. We can calculate that 
\begin{align*}
\frac{d}{ds}(a_1^2 - a_2^2) &= 2\psquare{a_1 \pround{\frac{1}{2} a_1  (b_1 - b_2)} - a_2 \pround{\frac{1}{2} a_2 ( b_2 - b_3)}} , \\
&= a_1^2 (b_1 - b_2) - a_2^2 (b_2 - b_3), \\
&= a_1^2 (b_1 - 2b_2 + b_3) + (a_1^2 - a_2^2)(b_2 - b_3) , \\
&= -3 b_2 a_1^2 + (a_1^2 - a_2^2)(b_2 - b_3). 
\end{align*}
Hence we can consistently set $a_1^2 - a_2^2 = 0$ provided $b_2 a_1^2 = 0$. As $b_2^\prime = a_2^2 - a_1^2$, this means we can consistently set $a_1^2 = a_2^2$ and $b_2 = 0$. Making these restrictions we can now eliminate the one remaining equation to find 
\[
0 = a_0^2 + 2 a_1^2 - \alpha_2 - d_1^2
\Rightarrow a_1^2 \pround{2\frac{da_1}{ds}}^2 = \beta^2 + 2 a_1^6 - \alpha_2 a_1^4 . 
\]

\section{Initial Computation of \secmath{V_4} Nahm Matrices}
\label{sec: v4 nahm matrices}
Taking the polynomials $\zeta_0 \zeta_1 (\zeta_0^4-\zeta_1^4)$, $\zeta_0^2 \zeta_1^2$, and $\zeta_0^4 +  \zeta_1^4$ as the inputs to the procedure of outlined in Appendix \ref{sec: initial computation of nahm matrices} gives the ODES in the six real-valued variables
\begin{align*}
    x^\prime &= 2 x^{2} - \frac{1}{6} y_{0}^{2} + \frac{1}{2} y_{1}^{2} - \frac{1}{2} y_{2}^{2} + \frac{1}{6} y_{3}^{2} - \frac{1}{2} y_{4}^{2},
    &y_2^\prime &= \frac{1}{3} y_{0}^{2} + y_{1}^{2} - 2 x y_{2} - y_{2}^{2} - \frac{1}{3} y_{3}^{2} - y_{1} y_{4} , \\
    y_0^\prime &= -2 x y_{0} + 2 y_{0} y_{2} + 2 y_{1} y_{3} + y_{3} y_{4}, 
    &y_3^\prime &= 2 y_{0} y_{1} + 2 x y_{3} - 2 y_{2} y_{3} + y_{0} y_{4},\\
    y_1^\prime &= 2 x y_{1} + 2 y_{1} y_{2} + \frac{2}{3} y_{0} y_{3} - y_{2} y_{4},
    &y_4^\prime &= -2 y_{1} y_{2} + \frac{2}{3} y_{0} y_{3} - 4 x y_{4},
\end{align*}
with the corresponding spectral curve 
\[
\mathcal{C}: \quad \eta^3 + \eta \psquare{a(\zeta^4 + 1) + b \zeta^2} + c\zeta (\zeta^4 - 1) = 0,
\]
where 
\begin{align*}
    a &= 8 x y_{0} + 4 y_{0} y_{2} - 4 y_{1} y_{3} + 4 y_{3} y_{4}, \\
    b &= 4 y_{0}^{2} - 12 y_{1}^{2} + 48 x y_{2} - 12 y_{2}^{2} + 4 y_{3}^{2} - 24 y_{1} y_{4}, \\
    c &= -8 i y_{0}^{2} y_{1} - 8 i y_{1}^{3} + 48 i x y_{1} y_{2} + 24 i y_{1} y_{2}^{2} - 16 i x y_{0} y_{3} + 16 i y_{0} y_{2} y_{3} \\
    &\phantom{=} + 8 i y_{1} y_{3}^{2} + 48 i x^{2} y_{4} - 4 i y_{0}^{2} y_{4} + 12 i y_{1}^{2} y_{4} - 12 i y_{2}^{2} y_{4} + 4 i y_{3}^{2} y_{4} - 4 i y_{4}^{3} . 
\end{align*}
The full Nahm matrices are 
\begin{align*}
    T_1 &= \begin{pmatrix}
        0 & 0 & 0 \\ 0 & 0 & -\bar{f_1} \\ 0 & f_1 & 0 
    \end{pmatrix}, &
    T_2 &= \begin{pmatrix}
        0 & 0 & f_2 \\ 0 & 0 & 0 \\ -\bar{f_2} & 0 & 0 
    \end{pmatrix}, &
    T_3 &= \begin{pmatrix}
        0 & -\bar{f_3} & 0 \\ f_3 & 0 & 0 \\ 0 & 0 & 0
    \end{pmatrix}.
\end{align*}
where the $f_i$ are given by 
\begin{align*}
f_1 &= 2x + y_0 - iy_1 + y_2 + iy_3 + iy_4, \\
f_2 &= 2x - y_0 - iy_1 + y_2 - iy_3 + iy_4, \\
f_3 &= 2x + 2iy_1 - 2y_2 + iy_4,
\end{align*}
One can check that setting $y_1 = y_3 = y_4 = 0$ is consistent, and corresponds to the inversion symmetric case. Note the condition on the residues of the Nahm data now
become that  the residue of each $f_i$ at the poles is $1$.

\providecommand{\bysame}{\leavevmode\hbox
to3em{\hrulefill}\thinspace}
\bibliographystyle{plain}
\bibliography{jabref_library.bib}

\end{document}